\documentclass[runningheads]{llncs}
\usepackage[dvips]{epsfig,graphicx}
\usepackage{subfigure}
\usepackage{varwidth}
\usepackage{psfrag}
\usepackage{complexity}
\usepackage{lineno} 
\usepackage{algorithmic}
\usepackage[ruled, algo2e, vlined, english]{algorithm2e}
\usepackage{algorithm}
\usepackage{amsmath}
\usepackage{float}
\usepackage[cal=boondoxo]{mathalfa}
\usepackage{mathrsfs}
\usepackage{tikz}
\usepackage{scalefnt}
\usepackage{caption}
\usepackage{xspace}
\usepackage{amsmath}
\usepackage{graphicx,url}
\usepackage{booktabs}
\usepackage{float}
\begin{document}
\title{Perfect matching cuts partitioning a graph into complementary subgraphs\thanks{This research has received funding from Rio de Janeiro Research Support Foundation (FAPERJ) under grant agreement E-26/201.344/2021,  National Council for Scientific  and Technological Development (CNPq) under grant agreement 309832/2020-9, and the European Research Council (ERC) under the European Union's Horizon $2020$ research and innovation programme under grant agreement CUTACOMBS (No. $714704$). }}

%
%
\author{
Diane Castonguay\inst{1} \and
Erika M. M. Coelho\inst{1} \and 
Hebert Coelho\inst{1} \and 
Julliano R. Nascimento\inst{1} \and
U\'{e}verton S. Souza\inst{2,3} 
}
\authorrunning{Castonguay et al.}
%
\institute{Instituto de Informática, Universidade Federal de Goi\'{a}s, Goiânia, Brazil \\
\email{ \{diane,erikamorais,hebert,jullianonascimento\}@inf.ufg.br} \and
Instituto de Computação, Universidade Federal Fluminense, Niterói, Brazil \and
Institute of Informatics, University of Warsaw, Warsaw, Poland \\
\email{ueverton@ic.uff.br} 
}
\maketitle              
\begin{abstract}
In \textsc{Partition Into Complementary Subgraphs (Comp-Sub)} we are given a graph $G=(V,E)$, and an edge set property $\Pi$, and asked whether $G$ can be decomposed into two graphs, $H$ and its complement $\overline{H}$, for some graph $H$, in such a way that the edge cut $[V(H),V(\overline{H})]$ satisfies the property $\Pi$.
Motivated by previous work, we consider $\textsc{Comp-Sub}(\Pi)$ when the property $\Pi=\mathcal{PM}$ specifies that the edge cut of the decomposition is a perfect matching.
We prove that $\textsc{Comp-Sub}(\mathcal{PM})$ is $\GI$-hard when the graph $G$ is 
$\{C_{k\geq 7}, \overline{C}_{k\geq 7} \}$-free. On the other hand, we show that $\textsc{Comp-Sub}(\mathcal{PM})$ is polynomial time solvable on $hole$-free graphs and on $P_5$-free graphs. Furthermore, we  present characterizations of $\textsc{Comp-Sub}(\mathcal{PM})$ on chordal, distance-hereditary, and extended $P_4$-laden graphs.

\end{abstract}
\section{Introduction}
\label{sec:introduction}

Finding graph partitions with some special properties has been a topic of extensive research.
Several combinatorial problems can be viewed as partition problems, such as {\sc Vertex Coloring} and {\sc Clique Cover}. 
In addition, many graph classes, e.g. bipartite and split graphs, can also be defined through a partition of its vertex set. In particular, the class of \emph{complementary prisms}~\cite{haynes2007complementary} are defined over complementary parts. 
The \textit{complementary prism} $G\overline{G}$ of a graph $G$ arises from the disjoint union of the graph $G$ and its complement $\overline{G}$ by adding the edges of a perfect matching between vertices with same label in $G$ and $\overline{G}$. Studies concerning the computational complexity of classical graph problems restricted to the class of complementary prisms graphs can be found in~\cite{Camargo_Prism,Duarte_Prism}. 

We say that a graph $G=(V,E)$ is decomposed into two graphs $G_1$ and $G_2$ if $V(G)$ can be partitioned into $V_1$ and $V_2$, where $G[V_1]=G_1$ and $G[V_2]=G_2$. The edge cut $[V_1,V_2]$ is called the edge cut of this decomposition.

As a generalization of complementary prisms, Nascimento, Souza and Szwarcfiter~\cite{nascimento2021partitioning} introduced the problem defined as follows.

\noindent
\fbox{\begin{minipage}{0.98\textwidth}
\textsc{Partition Into Complementary Subgraphs (Comp-Sub)}\vspace{.1cm}

\textbf{Instance:} A graph $G=(V,E)$, and an edge set property $\Pi$.\vspace{.1cm}

\textbf{Question:} Can $G$ be decomposed into two graphs, $H$ and its complement $\overline{H}$, for some graph $H$, in such a way that the edge cut $M$ of the decomposition satisfies the property $\Pi$?
\end{minipage}}

\medskip

For short, we abbreviate \textsc{Partition Into Complementary Subgraphs} with the edge set property $\Pi$  as $\textsc{Comp-Sub}(\Pi)$. We write $G\in \textsc{Comp-Sub}(\Pi)$ to denote that $G$ is a $yes$-instance of $\textsc{Comp-Sub}(\Pi)$ and we call $(H,\overline{H})$ as a \textit{complementary decomposition} of $G$.


The $\textsc{Comp-Sub}(\Pi)$ problem also finds motivation in parameterized complexity.
Recognizing whether a graph has a complementary decomposition can be useful for solving problems in \textsf{FPT}-time, as pointed out in~\cite{nascimento2021partitioning}. Nascimento, Souza and Szwarcfiter~\cite{nascimento2021partitioning} considered the cases where the edge cut $M$ is empty or induces a complete bipartite graph. They also presented some remarks when $\Pi$ is a general edge set property. In particular, when $M$ is empty, they make some links between $\textsc{Comp-Sub}(\Pi)$ and the \textsc{Graph Isomorphism} problem, from which they show that $\textsc{Comp-Sub}(\Pi)$ is $\GI$-hard. 

It is known that the recognition of complementary prisms can be done in polynomial time~\cite{cappelle2014recognizing}. This implies that, when the property $\Pi$ is a perfect matching $M$ between corresponding vertices in $H$ and $\overline{H}$, the $\textsc{Comp-Sub}(\Pi)$ problem is polynomial-time solvable.
So, a natural question is the study of $\textsc{Comp-Sub}(\Pi)$ when $\Pi$ specifies that $M$ is any perfect matching. In this context, two related problems arise: \textsc{Matching Cut}~\cite{kratsch2016algorithms,patrignani2001complexity} and \textsc{Perfect Matching Cut} \cite{heggernes1998partitioning}.
A \emph{(perfect) matching cut} is a partition of vertices of a graph into two parts such that the set of edges crossing between the parts forms a (perfect) matching. 
 Considering $\Pi=\mathcal{PM}$ as the property of being a perfect matching, $\textsc{Comp-Sub}(\mathcal{PM})$ can be seen as a variant of \textsc{Perfect Matching Cut} with the additional restriction that the two parts must induce complementary subgraphs. Note that studies regarding matchings satisfying particular constraints have received wide attention in the literature (c.f.~\cite{InducedMatchings,lima2021computational,IPL_LIMA,UniquelyRestricted,dmtcsProttiSouza}).

Motivated by Nascimento, Souza and Szwarcfiter~\cite{nascimento2021partitioning}, in this paper we deal with $\textsc{Comp-Sub}(\Pi)$, when $\Pi=\mathcal{PM}$ considers $M$ as a perfect matching. 
We show that $\textsc{Comp-Sub}(\mathcal{PM})$ is $\GI$-hard when the graph $G$ is 
$\{C_{k\geq 7}, \overline{C}_{k\geq 7} \}$-free. 
On the other hand, we present polynomial time algorithms able to solve $\textsc{Comp-Sub}(\mathcal{PM})$ when the input graph $G$ is $hole$-free or $P_5$-free. In addition, we characterize graphs $G \in \textsc{Comp-Sub}(\mathcal{PM})$ when $G$ is chordal,  distance-hereditary, or extended $P_4$-laden. Although extended $P_4$-laden graphs generalize cographs, we also show a simpler characterization for cographs.

The paper is organized as follows. Section~\ref{sec:preliminaries} contains some fundamental concepts and an auxiliary result. Sections~\ref{sec:results-cycle} and~\ref{sec:results-p4s} contains our results on some $cycle$-free graphs and graphs with few $P_4$'s, respectively. Further discussions are presented in Section~\ref{sec:conclusions}. 

\section{Preliminaries}
\label{sec:preliminaries}

We consider only finite, simple, and undirected graphs, and we use standard terminology and notation. See~\cite{bondy1976graph} for graph-theoretic terms not defined here.

Let $G$ be a graph. For a vertex $v \in V(G)$, we denote its \textit{open neighborhood} by $N_G(v)$, and its \textit{closed neighborhood}, denoted by $N_G[v] := N_G(v) \cup \{v\}$.  For a set $U \subseteq V(G)$, let $N_G(U) = \bigcup_{v \in U} N_G(v) \setminus U$, and $N_G[U] = N_G(U) \cup U$.  
The subgraph of $G$ \textit{induced} by $U$, denoted by $G[U]$, is the graph whose vertex set is $U$ and whose edge set consists of all the edges in $E(G)$ that have both endvertices in $U$.

Let $G$ be a graph. A set $U \subseteq V(G)$ is called a \textit{clique} (resp. \textit{independent set}) if the vertices in $U$ are pairwise adjacent (resp. nonadjacent).
We denote by $K_n$ a \textit{complete graph}, $I_n$ an independent set, $P_n$ a \textit{path graph}, and $C_n$ a \textit{cycle graph} on $n$ vertices.
Let $r$ be a positive integer. An $r$\emph{-partite graph} is one whose vertex set can be partitioned into $r$ subsets, in such a way that no edge has both ends in the same subset. An $r$-partite graph is \emph{complete} if any two vertices in different subsets are adjacent. When $r$ is not specified, we simply say \textit{(complete) multipartite}. 
A \textit{split graph} $G$ is one whose vertex set admits a partition $V(G) = C \cup I$ into a clique $C$ and an independent set $I$.
The \textit{complement} $\overline{G}$ of a graph $G$ is the graph defined by $V(\overline{G}) = V(G)$ and $uv \in E(\overline{G})$ if and only if $uv \notin E(G)$.


Let $P = v_1v_2\dots v_n$ be a path. We call $v_2, \dots, v_{n-1}$ as \textit{inner vertices} of $P$.
Two or more paths in a graph are \textit{independent} if none of them contains an inner vertex of another. A graph $G$ is $\ell$\textit{-connected} if any two of its vertices can be joined by $\ell$ independent paths. A $2$-connected graph is called {\em biconnected}.

A vertex $v$ in a graph $G$ is a \textit{cutvertex} or \textit{cutpoint}, if $G \setminus \{v\}$ is disconnected. A maximal connected subgraph without a cutpoint is a \textit{block}.
The \textit{block-cutpoint tree}  of a graph $G$ is a bipartite
graph whose vertex set consists of the set of cutpoints of $G$ and
the set of blocks of $G$. A cutpoint is adjacent to a block
whenever the cutpoint belongs to the block in $G$.

Two graphs $G=(V ,E)$ and $G'=(V',E')$ are \textit{isomorphic}, denoted as $G \simeq H$, if and only if there is a bijection, called \textit{isomorphism function}, $\varphi: V \to V'$ such that $uv \in E$ if and only if $\varphi(u)\varphi(v) \in E'$, for every $u,v \in V$. A graph $G$ is \textit{self-complementary} if $G \simeq \overline{G}$. The \textsc{Graph Isomorphism} problem receives as input two graphs $G$ and $G'$ and asks whether $G \simeq G'$.
We denote by $\GI$ the class of problems that admit a polynomial-time reduction to \textsc{Graph Isomorphism}. 

A problem $Q$ is $\GI$-\textit{complete} if the two conditions are satisfied:
(i) $Q$ is a member of $\GI$; and (ii) $Q$ is $\GI$-\textit{hard}, that is, for every problem $Q' \in \GI$, $Q'$ is polynomially reducible to $Q$.

We denote the set of positive integers $\{1, \dots, k\}$ by $[k]$. Let $G$ and $G_1,\dots, G_k$ be graphs. We say that $G$ is \textit{$\{G_1,\dots, G_k\}$-free} if $G$ does not contain $G_i$ as an induced subgraph, for every $i \in [k]$. 


Let $G$ and $H$ be two graphs such that $V(G)\cap V(H)=\emptyset$. The \textit{disjoint union} of $G$ and $H$, denoted by $G \cup H$, is the graph with $V(G \cup H) = V(G) \cup V(H)$ and $E(G \cup H) = E(G) \cup E(H)$.
The \textit{join} of $G$ and $H$, denoted by $G + H$, is the graph with $V(G + H) = V(G) \cup V(H)$ and $E(G+H) = E(G) \cup E(H) \cup \{uv : u \in V(G) \mbox{ and } v \in V(H)\}$. 

Let $G$ be a graph and $\mathcal{C}$ a class of graphs. A set $S \subseteq V(G)$ is a $\mathcal{C}$\textit{-modulator} if $G \setminus S$ belongs to $\mathcal{C}$. We define the \textit{distance} of $G$ to class $\mathcal{C}$ as the size of a minimum $S$ which is a $\mathcal{C}$-modulator.

Let $G$ be a graph that has a complementary decomposition $(G_1, G_2)$ with perfect matching cut $M = \{ u_1v_1, \dots, u_nv_n \}$, where $u_i \in V(G_1)$ and $v_i \in V(G_2)$, $i \in [n]$. We say that $u_i$ (resp. $v_i$) is the \textit{corresponding vertex} of $v_i$ (resp. $u_i$), for every $i \in [n]$. 
For $X \subseteq V(G_1)$, we call $X^{G_2} = \{ v_i \in V(G_2) : u_i \in X \}$ as the \textit{corresponding set} of $X$ over $G_2$. Similarly, for $X \subseteq V(G_2)$, we call $X^{G_1} = \{ u_i \in V(G_1) : v_i \in X \}$ as the \textit{corresponding set} of $X$ over $G_1$.

\medskip

Next, we present an auxiliary result, defined for $\textsc{Comp-Sub}(\mathcal{PM})$ with a restriction on the graphs of the decomposition. A \textit{cograph} is a $P_4$-free graph.

\begin{lemma}\label{lemma:H_in_cograph}
Let $G$ be a graph.
The problem of determining whether $G$ can be decomposed into two graphs, $G_1$, and its complement $G_2$, such that $G_1$ is a cograph and the edge cut of the decomposition is a perfect matching, can be solved in polynomial time.
\end{lemma}

\begin{proof}
Let $\mathcal{C}$ be the class of cographs and $G$ a $2n$-vertex graph. Suppose that $G$ is decomposable into complementary subgraphs $G_1$ and $G_2$, such that $G_1 \in \mathcal{C}$ and the edge cut $M$ of the decomposition is a perfect matching.

Since $\mathcal{C}$ is closed under complement, we have that $G_2 \in \mathcal{C}$. Given that \textsc{Graph Isomorphism} is linear-time solvable on cographs~\cite{colbourn1981on}, we perform a brute force algorithm to check every {\em relevant} partition $V(G_1), V(G_2)$ of $V(G)$. For that, we propose Algorithm~1, explained in sequel.

\begin{center}
\scalebox{0.95}{
\begin{minipage}{1.0\textwidth}
\begin{algorithm2e}[H]
\label{alg:cograph}
 \algsetup{linenosize=\small}
 \normalsize
 \DontPrintSemicolon
 \LinesNumbered
 \BlankLine
 \KwIn{A graph $G$.}
 \KwOut{Whether $G$ admits a complementary decomposition such that the edge cut of the decomposition is a perfect matching.}
 \BlankLine

    \ForAll{$x_1,x_2, y_1,y_2 \in V(G)$}{
        $V(G_1) := N_G[\{x_1,x_2\}] \setminus \{y_1,y_2\}$\;
        $V(G_2) := V(G) \setminus V(G_1)$\;
        $M := \{ xy \in E(G) : x \in V(G_1), y \in V(G_2)\}$\;
        \If{$M$ is a perfect matching \algorithmicand\ $G_1$ is a cograph \algorithmicand\ $G_2$ is a cograph \algorithmicand\ $G_1 \simeq \overline{G}_2$}{
            \Return yes\;
        }
    }
    \Return no\;
 
\caption{\textsc{Partition-Into-Complementary-Cographs}($G$)}
\end{algorithm2e}
\end{minipage}
}
\end{center}

We know that a cograph is connected if and only if its complement is disconnected~\cite{corneil1981complement}. Consequently, if a complementary decomposition $(G_1, G_2)$ exists, then either $G_1$ or $G_2$ is disconnected, say $G_2$. Then $G_1$ can be obtained by a join between the corresponding connected components of $G_2$. Thus, there exist two adjacent vertices $x_1, x_2 \in V(G_1)$, such that $N_{G_1}[\{x_1,x_2\}] = V(G_1)$. Furthermore, the edge set $M$ of the decomposition implies that there exist $y_1, y_2 \in V(G_2)$ such that $x_1y_1, x_2y_2 \in M$. 

By the above arguments, it is possible to find $V(G_1)$ by means of $N_{G}[\{x_1,x_2\}]$ except for two vertices $y_1, y_2 \in N_{G}[\{x_1,x_2\}]$ that must belong to $V(G_2)$. This way, $V(G_2)$ is obtained by $\{y_1,y_2\} \cup \{ v \in V(G) : v \notin N_{G}[\{x_1,x_2\}] \}$. 
Once found $V(G_1), V(G_2)$, and $M$, we test whether $M$ is a perfect matching and whether $G_1$ and $G_2$ are cographs. If so, we compute $\overline{G}_2$ and then we check isomorphism between $G_1$ and $\overline{G}_2$. 

The correctness of the algorithm follows from the fact that all the possible relevant partitions (for the emergence of the cographs, if any) are considered.

Now, we show that Algorithm~1 runs in polynomial time.

For enumerating every $4$-tuple of vertices $x_1,x_2, y_1,y_2 \in V(G)$ it is required $O(n^4)$ time. After, in $O(n+m)$ time we can check whether $M$ is a perfect matching, as well as checking whether $G_1$ and $G_2$ are cographs. Finally, for computing $\overline{G}_2$ and checking isomorphism between $G_1$ and $\overline{G}_2$ is also required $O(n+m)$ time. Therefore, the running time of Algorithm~1 takes $O(n^5+n^4m)$ time. 
\qed \end{proof}

\bigskip
\section{Results on some $C_k$-free graphs}
\label{sec:results-cycle}

We begin by showing a hardness result, in Theorem~\ref{theo:perfHardness}. 

\begin{theorem}\label{theo:perfHardness} $\textsc{Comp-Sub}(\mathcal{PM})$ is $\GI$-hard on $\{C_{k\geq 7}, \overline{C}_{k\geq 7} \}$-free graphs.
\end{theorem}

\begin{proof}
Given that \textsc{Graph Isomorphism} is $\GI$-hard on split graphs~\cite{chung1985on}, we show a polynomial-time reduction from such a problem to $\textsc{Comp-Sub}(\mathcal{PM})$. 

Note that a split graph is connected if and only if it does not contain isolated vertices. Therefore, we may assume that the instances of \textsc{Graph Isomorphism} on split graphs are pairs of connected split graphs.

Let $A$ and $B$ be connected split graphs such that $|V(A)| = |V(B)| = n$, for some $n \geq 3$.
From an instance $(A,B)$ of \textsc{Graph Isomorphism}, we construct an instance $G$ of $\textsc{Comp-Sub}(\mathcal{PM})$.

Let $G$ arise from the disjoint union between $A$, $\overline{B}$, $K_n$, and $I_n$. Denote $K_n$ by $K$ and $I_n$ by $I$.
We make every vertex in $V(A)$ adjacent to every vertex in $V(K)$. Furthermore, we add an arbitrary perfect matching between $V(A)$ and $V(I)$ and between $V(K)$ and $V(\overline{B})$.
An example of graph $G$ follows in Figure~\ref{fig:construcao_ueverton}.
Aditionally, let $H_1 = G[V(A) \cup V(K)]$ and $H_2 = G[V(\overline{B}) \cup V(I)]$. Clearly, the construction can be done in polynomial time.

We first show that $G$ is $\{C_{k\geq 7}, \overline{C}_{k\geq 7} \}$-free. 

\smallskip
\noindent{\bf Claim 1.}\label{claim:desired_class}
{\it
Let $G$ be the graph obtained from the construction. It holds that $G$ is a $\{C_{k\geq 7}, \overline{C}_{k\geq 7} \}$-free graph. 
}

\begin{proof} \textit{(of Claim~1.)}
We prove that (I) $G$ is $C_{k\geq 7}$-free, and (II) $G$ is $\overline{C}_{k\geq 7}$-free.

(I) Suppose by contradiction that $G$ contains a $C_{k\geq 7}$, denoted as $C$, as induced subgraph. We may assume that $k$ is minimum. 

By construction, $H_1$ and $H_2$ are split graphs and it is clear that $H_1$ and $H_2$ are $C_{\ell+4}$-free, for every $\ell \geq 0$. Then $V(C) \not\subseteq V(H_1)$ and $V(C) \not\subseteq V(H_2)$. So, we assume that $V(C) \cap V(H_1) \neq \emptyset$ and $V(C) \cap V(H_2) \neq \emptyset$. Since $I$ is a set of vertices with degree one in $G$, we have that $V(C) \cap I = \emptyset$. So, we may suppose that $V(C) \cap V(\overline{B}) \neq \emptyset$ and since $C$ is a cycle, $|V(C) \cap V(\overline{B})| \geq 2$. Since $\overline{B}$ is split, we have that $|V(C) \cap V(\overline{B})| \leq 4$.

Since  $C$ is a cycle and $K$ is a complete graph, $C$ must contain exactly two vertices from $K$ and no vertex of $A$. Then, $|V(C)| \geq 7$ implies that $|V(C)  \cap V(\overline{B})| \geq 5$, a contradiction.

\if 10

By construction 

Next, we define two subgraphs of $C$ and we show that their orders will lead to a contradiction.
Let $P = H_1[V(C) \cap V(H_1)]$ and $P' = H_2[V(C) \cap V(H_2)]$.

Since $C$ is an induced cycle in $G$, we have that $P$ is a disjoint union of induced paths in $H_1$. Let $P = (P^1 \cup \dots \cup P^t)$, for some $t \geq 1$. By construction, the edges in $M = \{a_ib_i, u_jv_j: 1 \leq i \leq n,  1 \leq j \leq 4 \}$ imply that $|V(P^q)| \geq 2$, for every $q \in [t]$. Furthermore, since $H_1$ is $P_4$-free, $|V(P^q)|  \leq 4$, for every $q \in [t]$. In addition, $H_1$ is $2K_2$-free, then we obtain that $t = 1$.

Thus, we conclude that $P \simeq P_{r}$, with $r \leq 4$, and with similar reasoning, we have that $P' \simeq P_{s}$, with $s \leq 4$.
Since $|V(C)| = r + s \leq 8$, we have a contradiction. Hence $G$ is $C_{k+9}$-free.

\fi

(II) Suppose by contradiction that $G$ contains a $\overline{C}_{k\geq 7}$, denoted as $D$, as induced subgraph. 
Let $V(D) = \{d_1, \dots, d_\ell\}$, for some $\ell \geq 7$, and $E(D) = \{d_id_j : 1 \leq i < j \leq \ell\} \setminus ( \{d_id_{i+1} : 1 \leq i \leq \ell-1 \} \cup \{d_\ell d_1\})$.

By definition of $D$, $\{d_1, d_2, d_4, d_5\}$ induces a $C_4$. Then, since $H_1$ and $H_2$ are split graphs, $V(D) \not\subseteq V(H_1)$ and $V(D) \not\subseteq V(H_2)$. So, we assume that $V(D) \cap V(H_1) \neq \emptyset$ and $V(D) \cap V(H_2) \neq \emptyset$. Then, there exists $i, j \in [\ell]$ such that $d_i \in V(H_1)$, $d_j \in V(H_2)$ and $d_id_j \in E(D)$.

Without loss of generality, suppose that $i = 1$. 
Since $\{d_1, d_3, d_5\}$ induces a $K_3$, we may assume that $\{d_1, d_3, d_5\} \subseteq V(H_1)$.
Thus, $d_1d_j \in E(D)$, for some $j \in \{4,6, \dots, \ell-1\}$. Notice that, if $j = 4$ (resp. $j \geq 6$), then $\{d_1, d_4, d_6\}$ (resp. $\{d_1, d_3, d_j\}$) induces a $K_3$ which intersects both $V(H_1)$ and $V(H_2)$, a contradiction. Therefore $G$ is $\overline{C}_{k\geq 7}$-free.\qed \end{proof}

In what follows, we prove that $(A,B)$ is a \emph{yes}-instance of \textsc{Graph Isomorphism} if and only if $G$ is a \emph{yes}-instance of $\textsc{Comp-Sub}(\mathcal{PM})$.

Suppose that $A \simeq B$. Since $I_n = \overline{K}_n$, $  \overline{\overline{B}} \simeq A$, and there is no edge between a vertex in $I$ and a vertex in $V(\overline{B})$, it is easy to see that $H_1$ and $\overline{H}_2$ are isomorphic. Therefore, $G$ is a \emph{yes}-instance of $\textsc{Comp-Sub}(\mathcal{PM})$.

For the converse, we suppose that $G$ is a \emph{yes}-instance of $\textsc{Comp-Sub}(\mathcal{PM})$. Let $(V',V'')$ be a partition of $V(G)$ into complementary parts such that $[V',V'']$ is a perfect matching. Since $I$ is a set of vertices with degree one in $G$ and $A$ is connected, it holds that either ($I\subset V'$ and $V(A)\subset V''$) or ($V(A)\subset V'$ and $I\subset V''$). Suppose that $V(A)\subset V'$. This implies that $V'=V(A)\cup K$ and $V''=V(\overline{B})\cup I$. Since $G[V']$ and $G[V'']$ are complementary, we have that $G[V'] \simeq \overline{G[V'']}$. Hence, due to the automorphism of universal vertices, it holds that $A\simeq B$.  
\qed \end{proof}

See in Figure~\ref{fig:construcao_ueverton} an example of the construction presented in Theorem~\ref{theo:perfHardness}.

\begin{figure}[htb]
    \centering
{\setlength{\fboxsep}{9pt}
\setlength{\fboxrule}{0.3pt}
\fbox{

\tikzset{every picture/.style={line width=0.5pt}} 

\begin{tikzpicture}[x=0.5pt,y=0.5pt,yscale=-1,xscale=1]

\draw [color={rgb, 255:red, 0; green, 0; blue, 0 }  ,draw opacity=1 ][fill={rgb, 255:red, 0; green, 0; blue, 0 }  ,fill opacity=1 ]   (198.27,63.09) -- (252.14,77.81) ;
\draw  [color={rgb, 255:red, 0; green, 0; blue, 0 }  ,draw opacity=1 ][fill={rgb, 255:red, 0; green, 0; blue, 0 }  ,fill opacity=1 ] (193.27,63.09) .. controls (193.27,60.33) and (195.51,58.09) .. (198.27,58.09) .. controls (201.03,58.09) and (203.27,60.33) .. (203.27,63.09) .. controls (203.27,65.85) and (201.03,68.09) .. (198.27,68.09) .. controls (195.51,68.09) and (193.27,65.85) .. (193.27,63.09) -- cycle ;
\draw  [color={rgb, 255:red, 0; green, 0; blue, 0 }  ,draw opacity=1 ][fill={rgb, 255:red, 0; green, 0; blue, 0 }  ,fill opacity=1 ] (215.35,112.2) .. controls (215.35,109.43) and (217.58,107.2) .. (220.35,107.2) .. controls (223.11,107.2) and (225.35,109.43) .. (225.35,112.2) .. controls (225.35,114.96) and (223.11,117.2) .. (220.35,117.2) .. controls (217.58,117.2) and (215.35,114.96) .. (215.35,112.2) -- cycle ;
\draw  [color={rgb, 255:red, 0; green, 0; blue, 0 }  ,draw opacity=1 ][fill={rgb, 255:red, 0; green, 0; blue, 0 }  ,fill opacity=1 ] (247.14,77.81) .. controls (247.14,75.05) and (249.38,72.81) .. (252.14,72.81) .. controls (254.9,72.81) and (257.14,75.05) .. (257.14,77.81) .. controls (257.14,80.57) and (254.9,82.81) .. (252.14,82.81) .. controls (249.38,82.81) and (247.14,80.57) .. (247.14,77.81) -- cycle ;
\draw  [color={rgb, 255:red, 0; green, 0; blue, 0 }  ,draw opacity=1 ][fill={rgb, 255:red, 0; green, 0; blue, 0 }  ,fill opacity=1 ] (337.8,64.86) .. controls (337.8,62.09) and (340.04,59.86) .. (342.8,59.86) .. controls (345.57,59.86) and (347.8,62.09) .. (347.8,64.86) .. controls (347.8,67.62) and (345.57,69.86) .. (342.8,69.86) .. controls (340.04,69.86) and (337.8,67.62) .. (337.8,64.86) -- cycle ;
\draw  [color={rgb, 255:red, 0; green, 0; blue, 0 }  ,draw opacity=1 ][fill={rgb, 255:red, 0; green, 0; blue, 0 }  ,fill opacity=1 ] (359.88,113.96) .. controls (359.88,111.2) and (362.12,108.96) .. (364.88,108.96) .. controls (367.64,108.96) and (369.88,111.2) .. (369.88,113.96) .. controls (369.88,116.72) and (367.64,118.96) .. (364.88,118.96) .. controls (362.12,118.96) and (359.88,116.72) .. (359.88,113.96) -- cycle ;
\draw  [color={rgb, 255:red, 0; green, 0; blue, 0 }  ,draw opacity=1 ][fill={rgb, 255:red, 0; green, 0; blue, 0 }  ,fill opacity=1 ] (391.68,79.58) .. controls (391.68,76.81) and (393.91,74.58) .. (396.68,74.58) .. controls (399.44,74.58) and (401.68,76.81) .. (401.68,79.58) .. controls (401.68,82.34) and (399.44,84.58) .. (396.68,84.58) .. controls (393.91,84.58) and (391.68,82.34) .. (391.68,79.58) -- cycle ;
\draw  [color={rgb, 255:red, 0; green, 0; blue, 0 }  ,draw opacity=1 ][fill={rgb, 255:red, 0; green, 0; blue, 0 }  ,fill opacity=1 ] (193.21,189.38) .. controls (193.21,186.62) and (195.45,184.38) .. (198.21,184.38) .. controls (200.97,184.38) and (203.21,186.62) .. (203.21,189.38) .. controls (203.21,192.15) and (200.97,194.38) .. (198.21,194.38) .. controls (195.45,194.38) and (193.21,192.15) .. (193.21,189.38) -- cycle ;
\draw  [color={rgb, 255:red, 0; green, 0; blue, 0 }  ,draw opacity=1 ][fill={rgb, 255:red, 0; green, 0; blue, 0 }  ,fill opacity=1 ] (215.28,238.49) .. controls (215.28,235.73) and (217.52,233.49) .. (220.28,233.49) .. controls (223.04,233.49) and (225.28,235.73) .. (225.28,238.49) .. controls (225.28,241.25) and (223.04,243.49) .. (220.28,243.49) .. controls (217.52,243.49) and (215.28,241.25) .. (215.28,238.49) -- cycle ;
\draw  [color={rgb, 255:red, 0; green, 0; blue, 0 }  ,draw opacity=1 ][fill={rgb, 255:red, 0; green, 0; blue, 0 }  ,fill opacity=1 ] (247.08,204.1) .. controls (247.08,201.34) and (249.32,199.1) .. (252.08,199.1) .. controls (254.84,199.1) and (257.08,201.34) .. (257.08,204.1) .. controls (257.08,206.86) and (254.84,209.1) .. (252.08,209.1) .. controls (249.32,209.1) and (247.08,206.86) .. (247.08,204.1) -- cycle ;
\draw  [color={rgb, 255:red, 0; green, 0; blue, 0 }  ,draw opacity=1 ][fill={rgb, 255:red, 0; green, 0; blue, 0 }  ,fill opacity=1 ] (334.24,187.74) .. controls (334.24,184.98) and (336.48,182.74) .. (339.24,182.74) .. controls (342,182.74) and (344.24,184.98) .. (344.24,187.74) .. controls (344.24,190.5) and (342,192.74) .. (339.24,192.74) .. controls (336.48,192.74) and (334.24,190.5) .. (334.24,187.74) -- cycle ;
\draw  [color={rgb, 255:red, 0; green, 0; blue, 0 }  ,draw opacity=1 ][fill={rgb, 255:red, 0; green, 0; blue, 0 }  ,fill opacity=1 ] (356.32,236.84) .. controls (356.32,234.08) and (358.55,231.84) .. (361.32,231.84) .. controls (364.08,231.84) and (366.32,234.08) .. (366.32,236.84) .. controls (366.32,239.61) and (364.08,241.84) .. (361.32,241.84) .. controls (358.55,241.84) and (356.32,239.61) .. (356.32,236.84) -- cycle ;
\draw  [color={rgb, 255:red, 0; green, 0; blue, 0 }  ,draw opacity=1 ][fill={rgb, 255:red, 0; green, 0; blue, 0 }  ,fill opacity=1 ] (388.11,202.46) .. controls (388.11,199.7) and (390.35,197.46) .. (393.11,197.46) .. controls (395.87,197.46) and (398.11,199.7) .. (398.11,202.46) .. controls (398.11,205.22) and (395.87,207.46) .. (393.11,207.46) .. controls (390.35,207.46) and (388.11,205.22) .. (388.11,202.46) -- cycle ;
\draw [color={rgb, 255:red, 0; green, 0; blue, 0 }  ,draw opacity=1 ][fill={rgb, 255:red, 0; green, 0; blue, 0 }  ,fill opacity=1 ]   (220.35,112.2) -- (252.14,77.81) ;
\draw [color={rgb, 255:red, 0; green, 0; blue, 0 }  ,draw opacity=1 ][fill={rgb, 255:red, 0; green, 0; blue, 0 }  ,fill opacity=1 ]   (198.21,189.38) -- (220.28,238.49) ;
\draw [color={rgb, 255:red, 0; green, 0; blue, 0 }  ,draw opacity=1 ][fill={rgb, 255:red, 0; green, 0; blue, 0 }  ,fill opacity=1 ]   (198.21,189.38) -- (252.08,204.1) ;
\draw [color={rgb, 255:red, 0; green, 0; blue, 0 }  ,draw opacity=1 ][fill={rgb, 255:red, 0; green, 0; blue, 0 }  ,fill opacity=1 ]   (220.28,238.49) -- (252.08,204.1) ;
\draw [color={rgb, 255:red, 0; green, 0; blue, 0 }  ,draw opacity=1 ][fill={rgb, 255:red, 0; green, 0; blue, 0 }  ,fill opacity=1 ]   (393.11,202.46) -- (361.32,236.84) ;
\draw [color={rgb, 255:red, 0; green, 0; blue, 0 }  ,draw opacity=1 ][fill={rgb, 255:red, 0; green, 0; blue, 0 }  ,fill opacity=1 ]   (198.21,189.38) -- (198.27,68.09) ;
\draw [color={rgb, 255:red, 0; green, 0; blue, 0 }  ,draw opacity=1 ][fill={rgb, 255:red, 0; green, 0; blue, 0 }  ,fill opacity=1 ]   (220.28,233.49) -- (220.35,112.2) ;
\draw [color={rgb, 255:red, 0; green, 0; blue, 0 }  ,draw opacity=1 ][fill={rgb, 255:red, 0; green, 0; blue, 0 }  ,fill opacity=1 ]   (252.08,204.1) -- (252.14,82.81) ;
\draw [color={rgb, 255:red, 0; green, 0; blue, 0 }  ,draw opacity=1 ][fill={rgb, 255:red, 0; green, 0; blue, 0 }  ,fill opacity=1 ]   (220.28,233.49) -- (198.27,63.09) ;
\draw [color={rgb, 255:red, 0; green, 0; blue, 0 }  ,draw opacity=1 ][fill={rgb, 255:red, 0; green, 0; blue, 0 }  ,fill opacity=1 ]   (220.28,233.49) -- (252.14,82.81) ;
\draw [color={rgb, 255:red, 0; green, 0; blue, 0 }  ,draw opacity=1 ][fill={rgb, 255:red, 0; green, 0; blue, 0 }  ,fill opacity=1 ]   (198.21,194.38) -- (220.35,112.2) ;
\draw [color={rgb, 255:red, 0; green, 0; blue, 0 }  ,draw opacity=1 ][fill={rgb, 255:red, 0; green, 0; blue, 0 }  ,fill opacity=1 ]   (252.08,204.1) -- (220.35,112.2) ;
\draw [color={rgb, 255:red, 0; green, 0; blue, 0 }  ,draw opacity=1 ][fill={rgb, 255:red, 0; green, 0; blue, 0 }  ,fill opacity=1 ]   (198.21,194.38) -- (252.14,77.81) ;
\draw [color={rgb, 255:red, 0; green, 0; blue, 0 }  ][line width=0.75] [line join = round][line cap = round]   (200.48,64.29) .. controls (205.37,64.97) and (208.47,70.26) .. (211.54,74.12) .. controls (229.95,97.17) and (238.16,107.74) .. (244.22,139.51) .. controls (250.9,174.57) and (251.53,206.17) .. (252.08,204.1) ;
\draw [color={rgb, 255:red, 0; green, 0; blue, 0 }  ,draw opacity=1 ][fill={rgb, 255:red, 0; green, 0; blue, 0 }  ,fill opacity=1 ]   (200.48,64.29) -- (342.8,64.86) ;
\draw [color={rgb, 255:red, 0; green, 0; blue, 0 }  ,draw opacity=1 ][fill={rgb, 255:red, 0; green, 0; blue, 0 }  ,fill opacity=1 ]   (254.35,79.01) -- (396.68,79.58) ;
\draw [color={rgb, 255:red, 0; green, 0; blue, 0 }  ,draw opacity=1 ][fill={rgb, 255:red, 0; green, 0; blue, 0 }  ,fill opacity=1 ]   (222.56,113.4) -- (364.88,113.96) ;
\draw [color={rgb, 255:red, 0; green, 0; blue, 0 }  ,draw opacity=1 ][fill={rgb, 255:red, 0; green, 0; blue, 0 }  ,fill opacity=1 ]   (198.21,189.38) -- (340.53,189.95) ;
\draw [color={rgb, 255:red, 0; green, 0; blue, 0 }  ,draw opacity=1 ][fill={rgb, 255:red, 0; green, 0; blue, 0 }  ,fill opacity=1 ]   (252.08,204.1) -- (394.4,204.67) ;
\draw [color={rgb, 255:red, 0; green, 0; blue, 0 }  ,draw opacity=1 ][fill={rgb, 255:red, 0; green, 0; blue, 0 }  ,fill opacity=1 ]   (220.28,238.49) -- (362.6,239.06) ;
\draw   (175.68,64.77) .. controls (175.68,55.73) and (183.01,48.4) .. (192.05,48.4) -- (259.43,48.4) .. controls (268.47,48.4) and (275.8,55.73) .. (275.8,64.77) -- (275.8,113.89) .. controls (275.8,122.93) and (268.47,130.26) .. (259.43,130.26) -- (192.05,130.26) .. controls (183.01,130.26) and (175.68,122.93) .. (175.68,113.89) -- cycle ;
\draw   (314.8,186.71) .. controls (314.8,177.75) and (322.06,170.49) .. (331.02,170.49) -- (397.8,170.49) .. controls (406.76,170.49) and (414.03,177.75) .. (414.03,186.71) -- (414.03,235.39) .. controls (414.03,244.35) and (406.76,251.61) .. (397.8,251.61) -- (331.02,251.61) .. controls (322.06,251.61) and (314.8,244.35) .. (314.8,235.39) -- cycle ;
\draw   (174.8,206.49) .. controls (174.8,179.32) and (196.82,157.3) .. (223.99,157.3) .. controls (251.16,157.3) and (273.19,179.32) .. (273.19,206.49) .. controls (273.19,233.66) and (251.16,255.69) .. (223.99,255.69) .. controls (196.82,255.69) and (174.8,233.66) .. (174.8,206.49) -- cycle ;
\draw   (315.8,88.01) .. controls (315.8,60.68) and (337.95,38.53) .. (365.28,38.53) .. controls (392.6,38.53) and (414.76,60.68) .. (414.76,88.01) .. controls (414.76,115.33) and (392.6,137.48) .. (365.28,137.48) .. controls (337.95,137.48) and (315.8,115.33) .. (315.8,88.01) -- cycle ;

\draw (147,76.61) node [anchor=north west][inner sep=0.75pt]    {$A$};
\draw (427.46,197.46) node [anchor=north west][inner sep=0.75pt]    {$\overline{B}$};
\draw (141.88,202.94) node [anchor=north west][inner sep=0.75pt]    {$K_{n}$};
\draw (423.44,69) node [anchor=north west][inner sep=0.75pt]    {$I_{n}$};
\draw (223,14.61) node [anchor=north west][inner sep=0.75pt]    {$H_{1}$};
\draw (359,15.4) node [anchor=north west][inner sep=0.75pt]    {$H_{2}$};
\draw (100,17.61) node [anchor=north west][inner sep=0.75pt]    {$G$};

\end{tikzpicture}

}}
    \caption{Graph $G$ constructed for Theorem~\ref{theo:perfHardness}.}
    \label{fig:construcao_ueverton}
\end{figure}
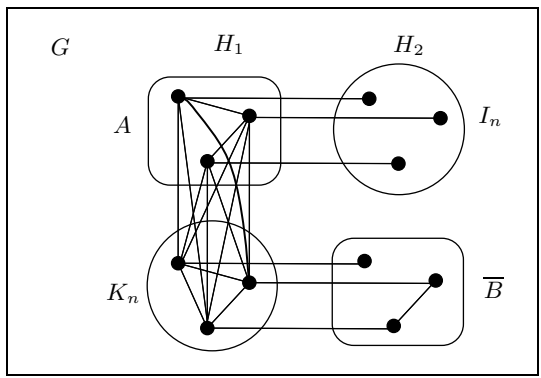

\if 10

Our next result shows that $\textsc{Comp-Sub}(\mathcal{PM})$ is $\GI$-hard for $C_5$-free graphs. Notice that the construction of $G$ presented in Theorem~\ref{theo:perfHardness} allows the existence of $C_5$ as induced subgraph. For instance, if $\overline{B}$ contains an induced $P_3 = v_1v_2v_3$, the corresponding vertices of $v_1$ and $v_3$ (in $K$) union with $v_1, v_2, v_3$ induce a $C_5$ in $G$. 
Hence, we adapt that construction to ensure a $C_5$-free graph.

\begin{theorem}\label{theo:perfHardnessC5} 
$\textsc{Comp-Sub}(\mathcal{PM})$ is $\GI$-hard on $C_5$-free graphs.
\end{theorem}

\if 10
\begin{proof} \textit{(Sketch)}
As in Theorem~\ref{theo:perfHardness}, we reduce \textsc{Graph Isomorphism} for split graphs to $\textsc{Comp-Sub}(\mathcal{PM})$. From an instance $A,B$ of \textsc{Graph Isomorphism}, we construct an instance $G$ of $\textsc{Comp-Sub}(\Pi)$ as follows.
\begin{itemize}

    \item Create a graph $H$ from the disjoint union of the graphs $A$, $\overline{B}$, $K_n$, $I_n$. 
    
    \item Add to $E(H)$ all the edges to make the following joins: $A + K_n$, $K_n + \overline{B}$, $\overline{B} + I_n$.

    \item Create a graph $H'$ by a copy of $H$. To distinguish vertices from $H$ and $H'$ we let $v' \in V(H')$ be the corresponding vertex of $v \in V(H)$. 
    
    \item Finally, let the graph $G$ arise from the disjoint union of $H$ and $H'$ by the addition of the edges $vv' \in E(G)$, for every $v \in V(H)$.
\end{itemize}

It is easy to see that $G$ can be obtained in polynomial time. Given that construction, it is possible to show that $G$ is a $C_5$-free graph and $A \simeq B$ if and only if $G$ admits a complementary decomposition $(H,H')$.
The complete proof is deferred to Appendix.
\qed \end{proof}

\fi

\begin{proof}
As in Theorem~\ref{theo:perfHardness}, we reduce \textsc{Graph Isomorphism} on split graphs to $\textsc{Comp-Sub}(\mathcal{PM})$. 

Let $A$ and $B$ be split graphs. We may assume that $A$ and $B$ does not contain isolated vertices and $|V(A)| = |V(B)| = n$, for some $n \geq 3$.
From an instance $A,B$ of \textsc{Graph Isomorphism}, we construct an instance $G$.

\begin{itemize}
    \item Create a graph $H$ from the disjoint union of the graphs $A$, $\overline{B}$, $K_n$, $I_n$. 
    
    \item Add to $E(H)$ all the edges to make the following joins: $A + K_n$, $K_n + \overline{B}$, $\overline{B} + I_n$,  Figure~\ref{fig:construcao2} contains an example of graph $H$. 
    
    \item Create a graph $H'$ by a copy of $H$. To distinguish vertices from $H$ and $H'$ we let $v' \in V(H')$ be the \textit{corresponding vertex} of $v \in V(H)$. 
    
    \item Finally, let the graph $G$ arise from the disjoint union between $H$ and $H'$ by the addition of the edges $vv' \in E(G)$, for every $v \in V(H)$.
\end{itemize}

Clearly the construction can be done in polynomial time. Next, we show that $G$ is a $C_5$-free graph. By contradiction, suppose that $G$ has a $C_5$, denoted as $C$, as an induced subgraph.

First, notice that, for every $u,v \in V(H)$, we have $uv \in E(H)$ if and only if $u'v' \in E(H')$. This implies that either $V(C) \subseteq V(H)$ or $V(C) \subseteq V(H')$. We may consider $V(C) \subseteq V(H)$, since the proof for $V(C) \subseteq V(H')$ is identical.
Furthermore, $A$, $\overline{B}$, $K_n$ and $I_n$ are $C_5$-free graphs, then we assume that $V(C)$ intersects at least two sets in $\{V(A), V(B), V(K_n), V(I_n) \}$.

Since $K_n$ is a complete graph, the joins $A + K_n $ and $K_n + \overline{B}$, imply that $|V(K_n) \cap C| \leq 1$. 
Similarly, since $I_n$ is an independent set, the join $\overline{B} + I_n$ implies that $|V(I_n) \cap C| \leq 1$.
Notice that $|V(K_n) \cap C| = 1$ implies that either $V(C) \subseteq V(A) \cup V(K_n)$ or $V(C) \subseteq V(K_n) \cup V(\overline{B}) \cup V(I_n)$.
In the former,  $|V(A) \cap C| \geq 4$, then the construction (recall the join $A + K_n$) implies that $C$ is not an induced cycle, a contradiction. In the latter, $|V(\overline{B}) \cap C| \geq 3$, and again, the construction (recall the join $K_n + \overline{B}$)
implies that $C$ is not an induced cycle, a contradiction. 
Finally, if $|V(K_n) \cap C| = 0$, we obtain that $V(C) \subseteq V(\overline{B}) \cup V(I_n)$  and $|V(\overline{B}) \cap C| \geq 4$. Again, the construction (recall the join $\overline{B} + I_n$)
implies that $C$ is not an induced cycle, a contradiction. Hence, $G$ is  $C_5$-free.

To show that $A \simeq B$ if and only if $G$ is a \emph{yes}-instance of $\textsc{Comp-Sub}(\Pi)$, it suffices to show that $H$ is self-complementary whenever $A \simeq B$. To this end, notice that $P_4$ is self-complementary,  $\overline{K}_n = I_n$, and $\overline{I}_n = K_n$. Furthermore, it is clear that  $\overline{A} = \overline{B}$ and  $\overline{\overline{B}} = A$ if and only if $A \simeq B$.
\qed \end{proof}

\begin{figure}
    \centering
{\setlength{\fboxsep}{9pt}
\setlength{\fboxrule}{0.3pt}
\fbox{

\tikzset{every picture/.style={line width=0.5pt}} 

\begin{tikzpicture}[x=0.44pt,y=0.44pt,yscale=-1,xscale=1]

\draw [color={rgb, 255:red, 0; green, 0; blue, 0 }  ,draw opacity=1 ][fill={rgb, 255:red, 0; green, 0; blue, 0 }  ,fill opacity=1 ]   (218.27,91.88) -- (272.14,106.6) ;
\draw  [color={rgb, 255:red, 0; green, 0; blue, 0 }  ,draw opacity=1 ][fill={rgb, 255:red, 0; green, 0; blue, 0 }  ,fill opacity=1 ] (213.27,91.88) .. controls (213.27,89.12) and (215.51,86.88) .. (218.27,86.88) .. controls (221.03,86.88) and (223.27,89.12) .. (223.27,91.88) .. controls (223.27,94.64) and (221.03,96.88) .. (218.27,96.88) .. controls (215.51,96.88) and (213.27,94.64) .. (213.27,91.88) -- cycle ;
\draw  [color={rgb, 255:red, 0; green, 0; blue, 0 }  ,draw opacity=1 ][fill={rgb, 255:red, 0; green, 0; blue, 0 }  ,fill opacity=1 ] (235.35,140.98) .. controls (235.35,138.22) and (237.58,135.98) .. (240.35,135.98) .. controls (243.11,135.98) and (245.35,138.22) .. (245.35,140.98) .. controls (245.35,143.75) and (243.11,145.98) .. (240.35,145.98) .. controls (237.58,145.98) and (235.35,143.75) .. (235.35,140.98) -- cycle ;
\draw  [color={rgb, 255:red, 0; green, 0; blue, 0 }  ,draw opacity=1 ][fill={rgb, 255:red, 0; green, 0; blue, 0 }  ,fill opacity=1 ] (267.14,106.6) .. controls (267.14,103.84) and (269.38,101.6) .. (272.14,101.6) .. controls (274.9,101.6) and (277.14,103.84) .. (277.14,106.6) .. controls (277.14,109.36) and (274.9,111.6) .. (272.14,111.6) .. controls (269.38,111.6) and (267.14,109.36) .. (267.14,106.6) -- cycle ;
\draw  [color={rgb, 255:red, 0; green, 0; blue, 0 }  ,draw opacity=1 ][fill={rgb, 255:red, 0; green, 0; blue, 0 }  ,fill opacity=1 ] (357.8,93.65) .. controls (357.8,90.88) and (360.04,88.65) .. (362.8,88.65) .. controls (365.57,88.65) and (367.8,90.88) .. (367.8,93.65) .. controls (367.8,96.41) and (365.57,98.65) .. (362.8,98.65) .. controls (360.04,98.65) and (357.8,96.41) .. (357.8,93.65) -- cycle ;
\draw  [color={rgb, 255:red, 0; green, 0; blue, 0 }  ,draw opacity=1 ][fill={rgb, 255:red, 0; green, 0; blue, 0 }  ,fill opacity=1 ] (379.88,142.75) .. controls (379.88,139.99) and (382.12,137.75) .. (384.88,137.75) .. controls (387.64,137.75) and (389.88,139.99) .. (389.88,142.75) .. controls (389.88,145.51) and (387.64,147.75) .. (384.88,147.75) .. controls (382.12,147.75) and (379.88,145.51) .. (379.88,142.75) -- cycle ;
\draw  [color={rgb, 255:red, 0; green, 0; blue, 0 }  ,draw opacity=1 ][fill={rgb, 255:red, 0; green, 0; blue, 0 }  ,fill opacity=1 ] (411.68,108.36) .. controls (411.68,105.6) and (413.91,103.36) .. (416.68,103.36) .. controls (419.44,103.36) and (421.68,105.6) .. (421.68,108.36) .. controls (421.68,111.13) and (419.44,113.36) .. (416.68,113.36) .. controls (413.91,113.36) and (411.68,111.13) .. (411.68,108.36) -- cycle ;
\draw  [color={rgb, 255:red, 0; green, 0; blue, 0 }  ,draw opacity=1 ][fill={rgb, 255:red, 0; green, 0; blue, 0 }  ,fill opacity=1 ] (213.21,218.17) .. controls (213.21,215.41) and (215.45,213.17) .. (218.21,213.17) .. controls (220.97,213.17) and (223.21,215.41) .. (223.21,218.17) .. controls (223.21,220.93) and (220.97,223.17) .. (218.21,223.17) .. controls (215.45,223.17) and (213.21,220.93) .. (213.21,218.17) -- cycle ;
\draw  [color={rgb, 255:red, 0; green, 0; blue, 0 }  ,draw opacity=1 ][fill={rgb, 255:red, 0; green, 0; blue, 0 }  ,fill opacity=1 ] (235.28,267.28) .. controls (235.28,264.52) and (237.52,262.28) .. (240.28,262.28) .. controls (243.04,262.28) and (245.28,264.52) .. (245.28,267.28) .. controls (245.28,270.04) and (243.04,272.28) .. (240.28,272.28) .. controls (237.52,272.28) and (235.28,270.04) .. (235.28,267.28) -- cycle ;
\draw  [color={rgb, 255:red, 0; green, 0; blue, 0 }  ,draw opacity=1 ][fill={rgb, 255:red, 0; green, 0; blue, 0 }  ,fill opacity=1 ] (267.08,232.89) .. controls (267.08,230.13) and (269.32,227.89) .. (272.08,227.89) .. controls (274.84,227.89) and (277.08,230.13) .. (277.08,232.89) .. controls (277.08,235.65) and (274.84,237.89) .. (272.08,237.89) .. controls (269.32,237.89) and (267.08,235.65) .. (267.08,232.89) -- cycle ;
\draw  [color={rgb, 255:red, 0; green, 0; blue, 0 }  ,draw opacity=1 ][fill={rgb, 255:red, 0; green, 0; blue, 0 }  ,fill opacity=1 ] (354.24,216.53) .. controls (354.24,213.77) and (356.48,211.53) .. (359.24,211.53) .. controls (362,211.53) and (364.24,213.77) .. (364.24,216.53) .. controls (364.24,219.29) and (362,221.53) .. (359.24,221.53) .. controls (356.48,221.53) and (354.24,219.29) .. (354.24,216.53) -- cycle ;
\draw  [color={rgb, 255:red, 0; green, 0; blue, 0 }  ,draw opacity=1 ][fill={rgb, 255:red, 0; green, 0; blue, 0 }  ,fill opacity=1 ] (376.32,265.63) .. controls (376.32,262.87) and (378.55,260.63) .. (381.32,260.63) .. controls (384.08,260.63) and (386.32,262.87) .. (386.32,265.63) .. controls (386.32,268.4) and (384.08,270.63) .. (381.32,270.63) .. controls (378.55,270.63) and (376.32,268.4) .. (376.32,265.63) -- cycle ;
\draw  [color={rgb, 255:red, 0; green, 0; blue, 0 }  ,draw opacity=1 ][fill={rgb, 255:red, 0; green, 0; blue, 0 }  ,fill opacity=1 ] (408.11,231.25) .. controls (408.11,228.49) and (410.35,226.25) .. (413.11,226.25) .. controls (415.87,226.25) and (418.11,228.49) .. (418.11,231.25) .. controls (418.11,234.01) and (415.87,236.25) .. (413.11,236.25) .. controls (410.35,236.25) and (408.11,234.01) .. (408.11,231.25) -- cycle ;
\draw [color={rgb, 255:red, 0; green, 0; blue, 0 }  ,draw opacity=1 ][fill={rgb, 255:red, 0; green, 0; blue, 0 }  ,fill opacity=1 ]   (240.35,140.98) -- (272.14,106.6) ;
\draw [color={rgb, 255:red, 0; green, 0; blue, 0 }  ,draw opacity=1 ][fill={rgb, 255:red, 0; green, 0; blue, 0 }  ,fill opacity=1 ]   (218.21,218.17) -- (240.28,267.28) ;
\draw [color={rgb, 255:red, 0; green, 0; blue, 0 }  ,draw opacity=1 ][fill={rgb, 255:red, 0; green, 0; blue, 0 }  ,fill opacity=1 ]   (218.21,218.17) -- (272.08,232.89) ;
\draw [color={rgb, 255:red, 0; green, 0; blue, 0 }  ,draw opacity=1 ][fill={rgb, 255:red, 0; green, 0; blue, 0 }  ,fill opacity=1 ]   (240.28,267.28) -- (272.08,232.89) ;
\draw [color={rgb, 255:red, 0; green, 0; blue, 0 }  ,draw opacity=1 ][fill={rgb, 255:red, 0; green, 0; blue, 0 }  ,fill opacity=1 ]   (413.11,231.25) -- (381.32,265.63) ;
\draw [color={rgb, 255:red, 0; green, 0; blue, 0 }  ,draw opacity=1 ][fill={rgb, 255:red, 0; green, 0; blue, 0 }  ,fill opacity=1 ]   (218.21,218.17) -- (218.27,96.88) ;
\draw [color={rgb, 255:red, 0; green, 0; blue, 0 }  ,draw opacity=1 ][fill={rgb, 255:red, 0; green, 0; blue, 0 }  ,fill opacity=1 ]   (240.28,262.28) -- (240.35,140.98) ;
\draw [color={rgb, 255:red, 0; green, 0; blue, 0 }  ,draw opacity=1 ][fill={rgb, 255:red, 0; green, 0; blue, 0 }  ,fill opacity=1 ]   (272.08,232.89) -- (272.14,111.6) ;
\draw [color={rgb, 255:red, 0; green, 0; blue, 0 }  ,draw opacity=1 ][fill={rgb, 255:red, 0; green, 0; blue, 0 }  ,fill opacity=1 ]   (240.28,262.28) -- (218.27,91.88) ;
\draw [color={rgb, 255:red, 0; green, 0; blue, 0 }  ,draw opacity=1 ][fill={rgb, 255:red, 0; green, 0; blue, 0 }  ,fill opacity=1 ]   (240.28,262.28) -- (272.14,111.6) ;
\draw [color={rgb, 255:red, 0; green, 0; blue, 0 }  ,draw opacity=1 ][fill={rgb, 255:red, 0; green, 0; blue, 0 }  ,fill opacity=1 ]   (218.21,223.17) -- (240.35,140.98) ;
\draw [color={rgb, 255:red, 0; green, 0; blue, 0 }  ,draw opacity=1 ][fill={rgb, 255:red, 0; green, 0; blue, 0 }  ,fill opacity=1 ]   (272.08,232.89) -- (240.35,140.98) ;
\draw [color={rgb, 255:red, 0; green, 0; blue, 0 }  ,draw opacity=1 ][fill={rgb, 255:red, 0; green, 0; blue, 0 }  ,fill opacity=1 ]   (218.21,223.17) -- (272.14,106.6) ;
\draw [color={rgb, 255:red, 0; green, 0; blue, 0 }  ][line width=0.75] [line join = round][line cap = round]   (220.48,93.08) .. controls (225.37,93.76) and (228.47,99.05) .. (231.54,102.91) .. controls (249.95,125.96) and (258.16,136.53) .. (264.22,168.3) .. controls (270.9,203.36) and (271.53,234.95) .. (272.08,232.89) ;
\draw [color={rgb, 255:red, 0; green, 0; blue, 0 }  ,draw opacity=1 ][fill={rgb, 255:red, 0; green, 0; blue, 0 }  ,fill opacity=1 ]   (218.21,218.17) -- (360.53,218.74) ;
\draw [color={rgb, 255:red, 0; green, 0; blue, 0 }  ,draw opacity=1 ][fill={rgb, 255:red, 0; green, 0; blue, 0 }  ,fill opacity=1 ]   (272.08,232.89) -- (414.4,233.46) ;
\draw [color={rgb, 255:red, 0; green, 0; blue, 0 }  ,draw opacity=1 ][fill={rgb, 255:red, 0; green, 0; blue, 0 }  ,fill opacity=1 ]   (240.28,267.28) -- (382.6,267.84) ;
\draw   (195.68,93.56) .. controls (195.68,84.52) and (203.01,77.19) .. (212.05,77.19) -- (279.43,77.19) .. controls (288.47,77.19) and (295.8,84.52) .. (295.8,93.56) -- (295.8,142.68) .. controls (295.8,151.72) and (288.47,159.05) .. (279.43,159.05) -- (212.05,159.05) .. controls (203.01,159.05) and (195.68,151.72) .. (195.68,142.68) -- cycle ;
\draw   (334.8,215.5) .. controls (334.8,206.54) and (342.06,199.28) .. (351.02,199.28) -- (417.8,199.28) .. controls (426.76,199.28) and (434.03,206.54) .. (434.03,215.5) -- (434.03,264.18) .. controls (434.03,273.14) and (426.76,280.4) .. (417.8,280.4) -- (351.02,280.4) .. controls (342.06,280.4) and (334.8,273.14) .. (334.8,264.18) -- cycle ;
\draw   (194.8,235.28) .. controls (194.8,208.11) and (216.82,186.09) .. (243.99,186.09) .. controls (271.16,186.09) and (293.19,208.11) .. (293.19,235.28) .. controls (293.19,262.45) and (271.16,284.47) .. (243.99,284.47) .. controls (216.82,284.47) and (194.8,262.45) .. (194.8,235.28) -- cycle ;
\draw   (338.47,118.13) .. controls (338.47,90.8) and (360.62,68.65) .. (387.94,68.65) .. controls (415.27,68.65) and (437.42,90.8) .. (437.42,118.13) .. controls (437.42,145.45) and (415.27,167.61) .. (387.94,167.61) .. controls (360.62,167.61) and (338.47,145.45) .. (338.47,118.13) -- cycle ;
\draw [color={rgb, 255:red, 0; green, 0; blue, 0 }  ,draw opacity=1 ][fill={rgb, 255:red, 0; green, 0; blue, 0 }  ,fill opacity=1 ]   (272.08,232.89) -- (360.53,218.74) ;
\draw [color={rgb, 255:red, 0; green, 0; blue, 0 }  ,draw opacity=1 ][fill={rgb, 255:red, 0; green, 0; blue, 0 }  ,fill opacity=1 ]   (245.28,267.28) -- (360.53,218.74) ;
\draw [color={rgb, 255:red, 0; green, 0; blue, 0 }  ,draw opacity=1 ][fill={rgb, 255:red, 0; green, 0; blue, 0 }  ,fill opacity=1 ]   (218.21,218.17) -- (413.11,231.25) ;
\draw [color={rgb, 255:red, 0; green, 0; blue, 0 }  ,draw opacity=1 ][fill={rgb, 255:red, 0; green, 0; blue, 0 }  ,fill opacity=1 ]   (240.28,267.28) -- (414.4,233.46) ;
\draw [color={rgb, 255:red, 0; green, 0; blue, 0 }  ,draw opacity=1 ][fill={rgb, 255:red, 0; green, 0; blue, 0 }  ,fill opacity=1 ]   (272.08,232.89) -- (381.32,265.63) ;
\draw    (218.21,218.17) .. controls (276.87,249.12) and (269.53,244.45) .. (381.32,265.63) ;
\draw [color={rgb, 255:red, 0; green, 0; blue, 0 }  ,draw opacity=1 ][fill={rgb, 255:red, 0; green, 0; blue, 0 }  ,fill opacity=1 ]   (361.54,216.84) -- (361.6,95.55) ;
\draw [color={rgb, 255:red, 0; green, 0; blue, 0 }  ,draw opacity=1 ][fill={rgb, 255:red, 0; green, 0; blue, 0 }  ,fill opacity=1 ]   (383.62,260.95) -- (383.68,139.65) ;
\draw [color={rgb, 255:red, 0; green, 0; blue, 0 }  ,draw opacity=1 ][fill={rgb, 255:red, 0; green, 0; blue, 0 }  ,fill opacity=1 ]   (415.41,231.56) -- (416.68,108.36) ;
\draw [color={rgb, 255:red, 0; green, 0; blue, 0 }  ,draw opacity=1 ][fill={rgb, 255:red, 0; green, 0; blue, 0 }  ,fill opacity=1 ]   (383.62,260.95) -- (361.6,90.55) ;
\draw [color={rgb, 255:red, 0; green, 0; blue, 0 }  ,draw opacity=1 ][fill={rgb, 255:red, 0; green, 0; blue, 0 }  ,fill opacity=1 ]   (383.62,260.95) -- (415.47,110.27) ;
\draw [color={rgb, 255:red, 0; green, 0; blue, 0 }  ,draw opacity=1 ][fill={rgb, 255:red, 0; green, 0; blue, 0 }  ,fill opacity=1 ]   (360.53,218.74) -- (383.68,139.65) ;
\draw [color={rgb, 255:red, 0; green, 0; blue, 0 }  ,draw opacity=1 ][fill={rgb, 255:red, 0; green, 0; blue, 0 }  ,fill opacity=1 ]   (415.41,231.56) -- (383.68,139.65) ;
\draw [color={rgb, 255:red, 0; green, 0; blue, 0 }  ,draw opacity=1 ][fill={rgb, 255:red, 0; green, 0; blue, 0 }  ,fill opacity=1 ]   (360.53,218.74) -- (415.47,105.27) ;
\draw [color={rgb, 255:red, 0; green, 0; blue, 0 }  ][line width=0.75] [line join = round][line cap = round]   (362.8,93.65) .. controls (391.19,120.35) and (407.85,171.68) .. (414.4,233.46) ;

\draw (167,105.4) node [anchor=north west][inner sep=0.75pt]    {$A$};
\draw (447.46,226.25) node [anchor=north west][inner sep=0.75pt]    {$\overline{B}$};
\draw (161.88,231.73) node [anchor=north west][inner sep=0.75pt]    {$K_{n}$};
\draw (443.44,97.78) node [anchor=north west][inner sep=0.75pt]    {$I_{n}$};
\draw (110.5,83.07) node [anchor=north west][inner sep=0.75pt]    {$H$};

\end{tikzpicture}

}}
    \caption{Graph $H$ constructed for Theorem~\ref{theo:perfHardnessC5}.}
    \label{fig:construcao2}
\end{figure}

\fi

Despite the hardness results presented in Theorem~\ref{theo:perfHardness}, 
next we show that $\textsc{Comp-Sub}(\mathcal{PM})$ can be solved in polynomial time on $hole$-free graphs. Recall that a $hole$ is a cycle on $5$ or more vertices.

\begin{theorem}\label{theo:holefree}
 $\textsc{Comp-Sub}(\mathcal{PM})$ is polynomial-time solvable on $hole$-free graphs.
\end{theorem}

\begin{proof}
Let $G$ be a $hole$-free graph having $2n$-vertices. 
We assume that $n$ is at least $5$; otherwise, the problem can be solved in $O(1)$ time.

Suppose that $G \in \textsc{Comp-Sub}(\mathcal{PM})$, then $G$ is decomposable into complementary subgraphs $G_1$ and $G_2$, such that the edge cut $M$ of the decomposition is a perfect matching.

Recall that $G_1$ or $G_2$ is a connected graph. Thus, we assume that $G_1$ is connected.

We go through the proof by analysing the structure of the graphs of the decomposition by means of their connectivity (Claims~1 and~2), and we conclude by showing how to find that decomposition when it exists.

\medskip
\noindent{\bf Claim 1.}\label{claim:2-conexo}
{\it
Let $G_1$ be a connected graph with at least five vertices and $F \subseteq V(G_1)$. If $G[F]$ is biconnected, then $G[F^{G_2}]$ is a cluster graph.
}
\begin{proof} \textit{(of Claim~1.)}
Suppose, by contradiction, that $G[F^{G_2}]$ is not a cluster graph and let $v_1 v_2 v_3$ be a $P_3$ in $G[F^{G_2}]$. Since $G[F]$ is $2$-connected, there exist two independent paths between any two vertices in $F$. 
Consider $u_1,u_2,u_3 \in F$ as the corresponding vertices of $v_1, v_2, v_3$, respectively.
Let $P$ and $P'$ be two independent paths between $u_1$ and $u_3$ in $F$. Since $P$ and $P'$ are independent, $u_2$ does not belong to $P$ or $P'$, say $P$. Then $P \cup \{v_1,v_2,v_3\}
$ induces a $hole$ in $G$, a contradiction. \qed \end{proof}

Next, we see more on the structure of $G_1$ and $G_2$.

\medskip
\noindent{\bf Claim 2.}\label{claim:1-conexo}
{\it
Let $G_1$ be a connected graph having at least five vertices.
If $G_1$ is non-biconnected, then 
either there is $S \subset V(G_1)$ with $|S|\leq 2$ such that $G_1  \setminus S$ is biconnected; 
or there is $S' \subset V(G_2)$ with $|S'|\leq 2$ such that $G_2 \setminus S'$ is biconnected.
}

\begin{proof} \textit{(of Claim~2.)}
Suppose that $G_1$ is non-biconnected and let $T$ be a block-cut-point tree of $G_1$. Let $\mathscr{B}= \{B_1, \dots, B_s\}$ and $\mathscr{C}= \{c_1, \dots, c_t\}$ be the sets of  blocks and cutpoints in $G_1$, respectively. 
The proof is divided in two cases: 
(I) $|\mathscr{B}| \geq 2$, $|\mathscr{C}| = 1$; and (II) $|\mathscr{B}| \geq 2$, $|\mathscr{C}| \geq 2$.

Recall that if $|\mathscr{B}| = 1$, then $|\mathscr{C}| = 0$ and $G_1$ is biconnected.

\smallskip
(I) Suppose that $|\mathscr{B}| \geq 2$ and $|\mathscr{C}| = 1$. 
Let $\mathscr{C}= \{c\}$. 
We have that 
$G_1 \setminus \{c\}$ is the disjoint union  $(B_1\setminus \{c\}) \cup \dots \cup (B_s \setminus \{c\})$. 
This implies that 
$\overline{G}_1 \setminus \{\overline{c}\}$ is the join  $(\overline{B}_1 \setminus \{\overline{c}\}) + \dots + (\overline{B}_s \setminus \{\overline{c}\})$. 

\begin{itemize}
    \item If $s\geq 3$ then $G_2 \setminus \{\overline{c}\} = (\overline{B}_1 \setminus \{\overline{c}\}) + \dots + (\overline{B}_s \setminus \{\overline{c}\})$ is biconnected.
    
    \item If $s=2$, $|\overline{B}_1 \setminus \{\overline{c}\}|\geq 2$, and $|\overline{B}_2 \setminus \{\overline{c}\}| \geq 2$ then $G_2 \setminus \{\overline{c}\} = (\overline{B}_1 \setminus \{\overline{c}\}) + (\overline{B}_2 \setminus \{\overline{c}\})$ is also biconnected.
    
    \item If $s=2$ and $|\overline{B}_1 \setminus \{\overline{c}\}|=1$ then $|\overline{B}_2 \setminus \{\overline{c}\}|\geq 2$. Otherwise, $G_1$ (and $G_2$) has only three vertices. Thus, ${B}_2$ is a block of $G_1$ with size $|V(G_1)|-1$, and $S={B}_1 \setminus \{{c}\}$ is as required. 
\end{itemize}



(II) Now, consider that  $|\mathscr{B}| \geq 2$ and $|\mathscr{C}| \geq 2$. Let $B,B' \in \mathscr{B}$ two distinct leaves in $T$ and $c,c' \in \mathscr{C}$ be two distinct cutpoints such that $Bc, B'c' \in E(T)$.

Let $D = V(G_1) \setminus (B \cup B')$.
Since $B$ (resp. $B'$) is a leaf in $T$, we have that $V(B) \setminus \{c\}$ (resp. $V(B') \setminus \{c'\}$) is not adjacent to $B' \cup D$ (resp. $B \cup D$). This implies that $\overline{G}_1 \setminus \{\overline{c},\overline{c}'\}$ is the join $(\overline{B} \setminus \{\overline{c}\}) + (\overline{B}' \setminus \{\overline{c}'\}) + \overline{D}$.

\begin{itemize}
    \item If $D \neq \emptyset$, we have that $(\overline{B} \setminus \{\overline{c}\}) + (\overline{B}' \setminus \{\overline{c}'\}) + \overline{D}$ is biconnected. Thus, ${G}_2 \setminus \{\overline{c},\overline{c}'\}$ is biconnected as required.
    

\item If $D = \emptyset$, $|B \setminus \{c\}| \geq 2$, and  $|B' \setminus \{c'\}| \geq 2$, then ${G}_2 \setminus \{\overline{c},\overline{c}'\} = (\overline{B} \setminus \{\overline{c}\}) + (\overline{B}' \setminus \{\overline{c}'\})$ is also biconnected.

\item If $D = \emptyset$ and $|B \setminus \{c\}| = 1$, then $|B' \setminus \{c'\}| \geq 2$. Otherwise, $G_1$ (and $G_2$) has only four vertices. Thus, $G_1 \setminus B$ is biconnected (notice that $|B| = 2$). 
\end{itemize}

This completes the proof of Claim~2.
\qed \end{proof}

By Claim~1, if $G_1$ is biconnected, then $G_2$ is a cluster graph. Since $G_1 \simeq \overline{G}_2$, we have that $G_1$ is a complete multipartite graph. Hence, $G_1$ and $G_2$ are cographs and, by Lemma~\ref{lemma:H_in_cograph}, we can find the complementary partition of $G$ in polynomial time.

Now, if $G_1$ is non-biconnected, recall that by Claim~2, 
either there is $S \subset V(G_1)$ with $|S|\leq 2$ such that $G_1  \setminus S$ is biconnected; 
or there is $S' \subset V(G_2)$ with $|S'|\leq 2$ such that $G_2 \setminus S'$ is biconnected.

Thus, there is a fixed number of vertices (at most $2$) such that removing from $G_1$ or $G_2$ leaves a biconnected graph. We deal with the case that there exist $c,c' \in V(G_2)$ such that $G_2 \setminus \{c,c'\}$ is biconnected. The approach for the other case is similar.

If there exist $c,c' \in V(G_2)$ such that $G_2 \setminus \{c,c'\}$ is $2$-connected, by Claim~1 (dual), we have that the graph induced by $(V(G_2) \setminus \{c,c'\})^{G_1}$ is a cluster graph. Then $G_1$ and $\overline{G}_2$ have distance to cluster equals $2$.  We proceed by Algorithm~2.

\begin{center}
\scalebox{0.95}{
\begin{minipage}{1.0\textwidth}
\begin{algorithm2e}[H]
\label{alg:dist_cluster}
 \algsetup{linenosize=\small}
 \normalsize
 \DontPrintSemicolon
 \LinesNumbered
 \BlankLine
 \KwIn{A graph $G$.}
 \KwOut{Whether $G$ is partitionable into two complementary graphs $G_1$ and $G_2$ such that $G_1$ and $\overline{G}_2$ have distance to cluster equals $2$ and the edge cut of the decomposition is a perfect matching.}
 \BlankLine

    \ForAll{$x_1, \dots, x_4, y_1, \dots, y_4 \in V(G)$}{
        $V(G_2) := N_G[\{y_1, \dots, y_4\}] \setminus \{x_1, \dots, x_4\}$\;
        $V(G_1) := V(G) \setminus V(G_2)$\;
        $M := \{ xy \in E(G) : x \in V(G_1), y \in V(G_2)\}$\;
        \If{$M$ is a perfect matching}{
                \ForAll{cluster-modulator $S_1$ of $G_1$, such that $|S_1| \leq 2$}{
                    \ForAll{cluster-modulator $S_2$ of $\overline{G}_2$, such that $|S_2| = |S_1|$}{
                        \ForAll{mapping $f:S_1\mapsto S_2$}{
                            \If{$f$ can be extended to an isomorphism from $G_1$ to $\overline{G}_2$}{
                                \Return yes\;
                            }
                        }
                     }
                }
        }
    }
    \Return no\;
 
\caption{\textsc{Partition-Into-Complementary-Subgraphs}($G$)}
\end{algorithm2e}
\end{minipage}
}
\end{center}

Since $G_2$ has distance to complete multipartite equals $2$, there exist four vertices $y_1, \dots, y_4 \in V(G_2)$ such that $N_{G_2}[\{y_1, \dots, y_4\}] = V(G_2)$. Then, if a complementary decomposition $(G_1,G_2)$ exists, we have that $|N_{G}[\{y_1, \dots, y_4\}]| = n + 4$. Thence, it is possible to find $V(G_2)$ which is $N_{G}[\{y_1, \dots, y_4\}]$ except for four vertices $x_1, \dots, x_4 \in N_{G}[\{y_1, \dots, y_4\}]$. We put $x_1, \dots, x_4$ in $V(G_1)$ as well as the remaining vertices $\{ v \in V(G) : v \notin N_{G}[\{y_1, \dots, y_4\}] \}$.  Given $V(G_1), V(G_2)$, and $M$, we check whether $M$ is a perfect matching. If so,  we compute $\overline{G}_2$ and we proceed to the step of finding  cluster-modulators $S_1$ for $G_1$ and $S_2$ for $\overline{G}_2$, that are done by Lines~6--7. In a naive manner, all the possible pair of modulators can be found in $O(n^4)$, but we show how to find them in a more efficient way.

We first find a $P_3 = w_1w_2w_3$ in $G_1$. We know that at least one vertex in $\{w_1,w_2,w_3\}$ must be included in a cluster-modulator for $G_1$. Then, for every $w \in \{w_1,w_2,w_3\}$ we put $w \in S_1$ and we branch by searching (if any) for a $P_3 = w'_1w'_2w'_3$ in $G_1 \setminus \{w\}$. Again, given that at least one vertex in $\{w'_1,w'_2,w'_3\}$ must be included in a cluster-modulator for $G_1$, for every $w' \in \{w_1,w_2,w_3\}$ we put $w' \in S_1$. If $G_1 \setminus S_1$ is a cluster graph, we proceed to finding, in the same manner, a cluster-modulator $S_2$ for $\overline{G}_2$. Note that this is basically a bounded search tree algorithm for finding cluster vertex deletion sets.

Given a pair of modulators $S_1$ and $S_2$ such that $|S_1|=|S_2|$, and a mapping from $S_1$ to $S_2$, the final task is checking if such a mapping can be extended to an isomorphism between $G_1$ and $\overline{G}_2$. Note that, by the bounded search tree technique, the number of pairs of modulators and mappings that must be considered is bounded by a constant.

Recall that $G_1$ (resp. $\overline{G}_2$) is a disjoint union of complete graphs $H_1 \cup \dots \cup H_p$ (resp. $H'_1 \cup \dots \cup  H'_p$), for some $p \geq 2$, with the addition of two vertices $w, w'$ (resp. $z,z'$) arbitrarily adjacent to $H_1 \cup \dots \cup  H_p$ (resp. $H'_1 \cup \dots \cup  H'_p$).
With this structure, an isomorphism from $G_1$ to $\overline{G}_2$ can be determined as follows. 

For a mapping $w \mapsto z$, $w' \mapsto z'$, we can map $H_i$ to $H'_j$, $i,j \in [p]$, if and only if 

\vspace{-0.2cm}
\begin{itemize}
\item[] $|V(H_i)| = |V(H'_j)|$ and 
\item[] $|N_{H_i}(w) \setminus N_{H_i}(w')| = |N_{H'_j}(z) \setminus N_{H'_j}(z')|$ and 
\item[] $|N_{H_i}(w') \setminus N_{H_i}(w)| = |N_{H'_j}(z') \setminus N_{H'_j}(z)|$ and 
\item[] $|N_{H_i}(w) \cap N_{H_j}(w')| = |N_{H'_i}(z) \cap N_{H'_j}(z')|$.
\end{itemize}

Therefore, each mapping $w \mapsto z$, $w' \mapsto z'$ defines ``types'' of cliques, from which the mapping can be extended to an isomorphism from $G_1$ to $\overline{G}_2$ if and only if $G_1$ and $\overline{G}_2$ have the same number of cliques per type.

Next, we analyse the running time of Algorithm~2.

First, in Line~1, we check every $8$-tuple of vertices in $V(G)$ to separate those $x_1, \dots, x_4 \in V(G_1)$ and $y_1, \dots, y_4 \in V(G_2)$, which requires $O(n^8)$ time. Lines~2--4 define $V(G_1), V(G_2)$, and $M$, which run in $O(n+m)$ time. Checking whether $M$ is a perfect matching (Line~5) can be done in $O(n+m)$ time. 

Recall that a $P_3$ in $G$ can be found in $O(n+m)$ time.
By the method previously described, Line~6 can be done by finding a $P_3 = w_1w_2w_3$ in $G_1$; for every $w \in \{w_1,w_2,w_3\}$ finding a $P_3 = w'_1w'_2w'_3$ in $G_1 \setminus \{w\}$ in $G_1$; and finally, for every $w' \in \{w_1,w_2,w_3\}$, checking whether $G_1 \setminus S_1 = \{w, w'\}$ is a cluster graph. This produces a ternary search tree with height equals $2$. Hence with $9$ leaf nodes, that are at most $9$ possible cluster-modulators $\{w,w'\}$ for $G_1$. This requires a running time of $O(m+n)$. For every of those possible cluster-modulators for $G_1$ we proceed to finding every cluster-modulator for $\overline{G}_2$ (Line~7) by the same method.
This gives an amount of at most $81$ possible $4$-tuples $w,w',z,z'$ that must be checked, hence Lines~6--8 run in $O(n+m)$ time.

Finally, for Line~9, checking whether an isomorphism from $G_1$ to  $\overline{G}_2$ can be extended from $f$ can be done by checking sizes of cliques and neighborhoods, which can be done in $O(n+m)$ time.

Therefore, the overall running time of Algorithm~2 is of order $O(n^8(n+m)) = O(n^9+n^8m)$. 
\qed \end{proof}

Next, it follows a characterization for $\textsc{Comp-Sub}(\mathcal{PM})$ in the class of distance hereditary graphs, which is a subclass of $hole$-free graphs. A \textit{distance-hereditary} graph is a $\{domino, house, gem, hole\}$-free graph. See a $domino$, a $house$ and a $gem$ in Figure~\ref{fig:small_subgraphs}. 
For the next result, let \Large{$\varrho$} \normalsize be the graph in Figure~\ref{fig:small_subgraphs}.

\begin{figure}[htb]
\centering
{\setlength{\fboxsep}{3.5pt}
\setlength{\fboxrule}{0.3pt}
\fbox{

\tikzset{every picture/.style={line width=0.5pt}} 

\begin{tikzpicture}[x=0.35pt,y=0.35pt,yscale=-1,xscale=1]

\draw  [color={rgb, 255:red, 0; green, 0; blue, 0 }  ,draw opacity=1 ][fill={rgb, 255:red, 0; green, 0; blue, 0 }  ,fill opacity=1 ] (109.33,118.8) .. controls (109.33,116.04) and (111.57,113.8) .. (114.33,113.8) .. controls (117.09,113.8) and (119.33,116.04) .. (119.33,118.8) .. controls (119.33,121.56) and (117.09,123.8) .. (114.33,123.8) .. controls (111.57,123.8) and (109.33,121.56) .. (109.33,118.8) -- cycle ;
\draw  [color={rgb, 255:red, 0; green, 0; blue, 0 }  ,draw opacity=1 ][fill={rgb, 255:red, 0; green, 0; blue, 0 }  ,fill opacity=1 ] (65.23,161.83) .. controls (65.23,159.07) and (67.47,156.83) .. (70.23,156.83) .. controls (72.99,156.83) and (75.23,159.07) .. (75.23,161.83) .. controls (75.23,164.59) and (72.99,166.83) .. (70.23,166.83) .. controls (67.47,166.83) and (65.23,164.59) .. (65.23,161.83) -- cycle ;
\draw    (70.23,118.8) -- (114.33,118.8) ;
\draw    (114.33,113.8) -- (114.33,161.83) ;
\draw    (114.33,161.83) -- (70.23,161.83) ;
\draw  [color={rgb, 255:red, 0; green, 0; blue, 0 }  ,draw opacity=1 ][fill={rgb, 255:red, 0; green, 0; blue, 0 }  ,fill opacity=1 ] (109.33,161.83) .. controls (109.33,159.07) and (111.57,156.83) .. (114.33,156.83) .. controls (117.09,156.83) and (119.33,159.07) .. (119.33,161.83) .. controls (119.33,164.59) and (117.09,166.83) .. (114.33,166.83) .. controls (111.57,166.83) and (109.33,164.59) .. (109.33,161.83) -- cycle ;
\draw  [color={rgb, 255:red, 0; green, 0; blue, 0 }  ,draw opacity=1 ][fill={rgb, 255:red, 0; green, 0; blue, 0 }  ,fill opacity=1 ] (65.23,118.8) .. controls (65.23,116.04) and (67.47,113.8) .. (70.23,113.8) .. controls (72.99,113.8) and (75.23,116.04) .. (75.23,118.8) .. controls (75.23,121.56) and (72.99,123.8) .. (70.23,123.8) .. controls (67.47,123.8) and (65.23,121.56) .. (65.23,118.8) -- cycle ;
\draw    (70.23,113.8) -- (70.23,161.83) ;
\draw  [color={rgb, 255:red, 0; green, 0; blue, 0 }  ,draw opacity=1 ][fill={rgb, 255:red, 0; green, 0; blue, 0 }  ,fill opacity=1 ] (222.33,120.13) .. controls (222.33,117.37) and (224.57,115.13) .. (227.33,115.13) .. controls (230.09,115.13) and (232.33,117.37) .. (232.33,120.13) .. controls (232.33,122.89) and (230.09,125.13) .. (227.33,125.13) .. controls (224.57,125.13) and (222.33,122.89) .. (222.33,120.13) -- cycle ;
\draw  [color={rgb, 255:red, 0; green, 0; blue, 0 }  ,draw opacity=1 ][fill={rgb, 255:red, 0; green, 0; blue, 0 }  ,fill opacity=1 ] (178.23,161.17) .. controls (178.23,158.41) and (180.47,156.17) .. (183.23,156.17) .. controls (185.99,156.17) and (188.23,158.41) .. (188.23,161.17) .. controls (188.23,163.93) and (185.99,166.17) .. (183.23,166.17) .. controls (180.47,166.17) and (178.23,163.93) .. (178.23,161.17) -- cycle ;
\draw    (183.23,120.13) -- (227.33,120.13) ;
\draw    (227.33,115.13) -- (227.33,161.17) ;
\draw    (227.33,161.17) -- (183.23,161.17) ;
\draw  [color={rgb, 255:red, 0; green, 0; blue, 0 }  ,draw opacity=1 ][fill={rgb, 255:red, 0; green, 0; blue, 0 }  ,fill opacity=1 ] (222.33,161.17) .. controls (222.33,158.41) and (224.57,156.17) .. (227.33,156.17) .. controls (230.09,156.17) and (232.33,158.41) .. (232.33,161.17) .. controls (232.33,163.93) and (230.09,166.17) .. (227.33,166.17) .. controls (224.57,166.17) and (222.33,163.93) .. (222.33,161.17) -- cycle ;
\draw  [color={rgb, 255:red, 0; green, 0; blue, 0 }  ,draw opacity=1 ][fill={rgb, 255:red, 0; green, 0; blue, 0 }  ,fill opacity=1 ] (178.23,120.13) .. controls (178.23,117.37) and (180.47,115.13) .. (183.23,115.13) .. controls (185.99,115.13) and (188.23,117.37) .. (188.23,120.13) .. controls (188.23,122.89) and (185.99,125.13) .. (183.23,125.13) .. controls (180.47,125.13) and (178.23,122.89) .. (178.23,120.13) -- cycle ;
\draw    (183.23,115.13) -- (183.23,166.17) ;
\draw  [color={rgb, 255:red, 0; green, 0; blue, 0 }  ,draw opacity=1 ][fill={rgb, 255:red, 0; green, 0; blue, 0 }  ,fill opacity=1 ] (222.33,79.13) .. controls (222.33,76.37) and (224.57,74.13) .. (227.33,74.13) .. controls (230.09,74.13) and (232.33,76.37) .. (232.33,79.13) .. controls (232.33,81.89) and (230.09,84.13) .. (227.33,84.13) .. controls (224.57,84.13) and (222.33,81.89) .. (222.33,79.13) -- cycle ;
\draw    (183.23,79.13) -- (227.33,79.13) ;
\draw    (227.33,79.13) -- (227.33,125.13) ;
\draw  [color={rgb, 255:red, 0; green, 0; blue, 0 }  ,draw opacity=1 ][fill={rgb, 255:red, 0; green, 0; blue, 0 }  ,fill opacity=1 ] (178.23,79.13) .. controls (178.23,76.37) and (180.47,74.13) .. (183.23,74.13) .. controls (185.99,74.13) and (188.23,76.37) .. (188.23,79.13) .. controls (188.23,81.89) and (185.99,84.13) .. (183.23,84.13) .. controls (180.47,84.13) and (178.23,81.89) .. (178.23,79.13) -- cycle ;
\draw    (183.23,79.13) -- (183.23,125.13) ;
\draw  [color={rgb, 255:red, 0; green, 0; blue, 0 }  ,draw opacity=1 ][fill={rgb, 255:red, 0; green, 0; blue, 0 }  ,fill opacity=1 ] (86.57,79.13) .. controls (86.57,76.37) and (88.81,74.13) .. (91.57,74.13) .. controls (94.33,74.13) and (96.57,76.37) .. (96.57,79.13) .. controls (96.57,81.89) and (94.33,84.13) .. (91.57,84.13) .. controls (88.81,84.13) and (86.57,81.89) .. (86.57,79.13) -- cycle ;
\draw    (91.57,79.13) -- (70.23,121.8) ;
\draw    (91.57,79.13) -- (114.33,121.8) ;
\draw  [color={rgb, 255:red, 0; green, 0; blue, 0 }  ,draw opacity=1 ][fill={rgb, 255:red, 0; green, 0; blue, 0 }  ,fill opacity=1 ] (332.67,94.8) .. controls (332.67,92.04) and (334.91,89.8) .. (337.67,89.8) .. controls (340.43,89.8) and (342.67,92.04) .. (342.67,94.8) .. controls (342.67,97.56) and (340.43,99.8) .. (337.67,99.8) .. controls (334.91,99.8) and (332.67,97.56) .. (332.67,94.8) -- cycle ;
\draw  [color={rgb, 255:red, 0; green, 0; blue, 0 }  ,draw opacity=1 ][fill={rgb, 255:red, 0; green, 0; blue, 0 }  ,fill opacity=1 ] (271.23,131.5) .. controls (271.23,128.74) and (273.47,126.5) .. (276.23,126.5) .. controls (278.99,126.5) and (281.23,128.74) .. (281.23,131.5) .. controls (281.23,134.26) and (278.99,136.5) .. (276.23,136.5) .. controls (273.47,136.5) and (271.23,134.26) .. (271.23,131.5) -- cycle ;
\draw    (293.57,94.8) -- (337.67,94.8) ;
\draw    (337.67,94.8) -- (357.33,131.5) ;
\draw  [color={rgb, 255:red, 0; green, 0; blue, 0 }  ,draw opacity=1 ][fill={rgb, 255:red, 0; green, 0; blue, 0 }  ,fill opacity=1 ] (352.33,131.5) .. controls (352.33,128.74) and (354.57,126.5) .. (357.33,126.5) .. controls (360.09,126.5) and (362.33,128.74) .. (362.33,131.5) .. controls (362.33,134.26) and (360.09,136.5) .. (357.33,136.5) .. controls (354.57,136.5) and (352.33,134.26) .. (352.33,131.5) -- cycle ;
\draw  [color={rgb, 255:red, 0; green, 0; blue, 0 }  ,draw opacity=1 ][fill={rgb, 255:red, 0; green, 0; blue, 0 }  ,fill opacity=1 ] (288.57,94.8) .. controls (288.57,92.04) and (290.81,89.8) .. (293.57,89.8) .. controls (296.33,89.8) and (298.57,92.04) .. (298.57,94.8) .. controls (298.57,97.56) and (296.33,99.8) .. (293.57,99.8) .. controls (290.81,99.8) and (288.57,97.56) .. (288.57,94.8) -- cycle ;
\draw    (293.57,94.8) -- (276.23,131.5) ;
\draw  [color={rgb, 255:red, 0; green, 0; blue, 0 }  ,draw opacity=1 ][fill={rgb, 255:red, 0; green, 0; blue, 0 }  ,fill opacity=1 ] (312.67,163.83) .. controls (312.67,161.07) and (314.91,158.83) .. (317.67,158.83) .. controls (320.43,158.83) and (322.67,161.07) .. (322.67,163.83) .. controls (322.67,166.59) and (320.43,168.83) .. (317.67,168.83) .. controls (314.91,168.83) and (312.67,166.59) .. (312.67,163.83) -- cycle ;
\draw    (357.33,131.5) -- (317.67,163.83) ;
\draw    (276.23,131.5) -- (317.67,163.83) ;
\draw    (337.67,91.8) -- (317.67,163.83) ;
\draw    (293.57,91.8) -- (317.67,163.83) ;
\draw  [color={rgb, 255:red, 0; green, 0; blue, 0 }  ,draw opacity=1 ][fill={rgb, 255:red, 0; green, 0; blue, 0 }  ,fill opacity=1 ] (449.67,119.83) .. controls (449.67,117.07) and (451.91,114.83) .. (454.67,114.83) .. controls (457.43,114.83) and (459.67,117.07) .. (459.67,119.83) .. controls (459.67,122.59) and (457.43,124.83) .. (454.67,124.83) .. controls (451.91,124.83) and (449.67,122.59) .. (449.67,119.83) -- cycle ;
\draw  [color={rgb, 255:red, 0; green, 0; blue, 0 }  ,draw opacity=1 ][fill={rgb, 255:red, 0; green, 0; blue, 0 }  ,fill opacity=1 ] (405.57,160.87) .. controls (405.57,158.11) and (407.81,155.87) .. (410.57,155.87) .. controls (413.33,155.87) and (415.57,158.11) .. (415.57,160.87) .. controls (415.57,163.63) and (413.33,165.87) .. (410.57,165.87) .. controls (407.81,165.87) and (405.57,163.63) .. (405.57,160.87) -- cycle ;
\draw    (410.57,119.83) -- (454.67,119.83) ;
\draw    (454.67,160.87) -- (410.57,160.87) ;
\draw  [color={rgb, 255:red, 0; green, 0; blue, 0 }  ,draw opacity=1 ][fill={rgb, 255:red, 0; green, 0; blue, 0 }  ,fill opacity=1 ] (449.67,160.87) .. controls (449.67,158.11) and (451.91,155.87) .. (454.67,155.87) .. controls (457.43,155.87) and (459.67,158.11) .. (459.67,160.87) .. controls (459.67,163.63) and (457.43,165.87) .. (454.67,165.87) .. controls (451.91,165.87) and (449.67,163.63) .. (449.67,160.87) -- cycle ;
\draw  [color={rgb, 255:red, 0; green, 0; blue, 0 }  ,draw opacity=1 ][fill={rgb, 255:red, 0; green, 0; blue, 0 }  ,fill opacity=1 ] (405.57,119.83) .. controls (405.57,117.07) and (407.81,114.83) .. (410.57,114.83) .. controls (413.33,114.83) and (415.57,117.07) .. (415.57,119.83) .. controls (415.57,122.59) and (413.33,124.83) .. (410.57,124.83) .. controls (407.81,124.83) and (405.57,122.59) .. (405.57,119.83) -- cycle ;
\draw    (410.57,114.83) -- (410.57,165.87) ;
\draw  [color={rgb, 255:red, 0; green, 0; blue, 0 }  ,draw opacity=1 ][fill={rgb, 255:red, 0; green, 0; blue, 0 }  ,fill opacity=1 ] (449.67,78.83) .. controls (449.67,76.07) and (451.91,73.83) .. (454.67,73.83) .. controls (457.43,73.83) and (459.67,76.07) .. (459.67,78.83) .. controls (459.67,81.59) and (457.43,83.83) .. (454.67,83.83) .. controls (451.91,83.83) and (449.67,81.59) .. (449.67,78.83) -- cycle ;
\draw    (410.57,78.83) -- (454.67,78.83) ;
\draw    (454.67,78.83) -- (454.67,124.83) ;
\draw  [color={rgb, 255:red, 0; green, 0; blue, 0 }  ,draw opacity=1 ][fill={rgb, 255:red, 0; green, 0; blue, 0 }  ,fill opacity=1 ] (405.57,78.83) .. controls (405.57,76.07) and (407.81,73.83) .. (410.57,73.83) .. controls (413.33,73.83) and (415.57,76.07) .. (415.57,78.83) .. controls (415.57,81.59) and (413.33,83.83) .. (410.57,83.83) .. controls (407.81,83.83) and (405.57,81.59) .. (405.57,78.83) -- cycle ;
\draw    (410.57,78.83) -- (410.57,124.83) ;

\draw (61.83,189.7) node [anchor=north west][inner sep=0.75pt]    {$house$};
\draw (158.17,189.37) node [anchor=north west][inner sep=0.75pt]    {$domino$};
\draw (290.33,189.2) node [anchor=north west][inner sep=0.75pt]    {$gem$};
\draw (418.33,180.1) node [anchor=north west][inner sep=0.75pt]  [font=\large]  {$\varrho $};

\end{tikzpicture}

}}
\caption{Some small subgraphs.}
\label{fig:small_subgraphs}
\end{figure}
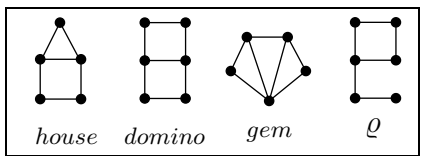

\begin{proposition}\label{prop:distanceHereditaryMPerfect}
Let $G$ be a distance-hereditary graph of order $2n$.
It holds that $G \in \textsc{Comp-Sub}(\mathcal{PM})$ if and only if $G \in \{K_n\overline{K}_n,$\Large{$\varrho$}\normalsize$\}$.
\end{proposition}

We close this section with a characterization of $\textsc{Comp-Sub}(\mathcal{PM})$ on chordal graphs. Recall that a chordal graph is a $C_{k\geq 4}$-free graph.

\begin{proposition}\label{prop:chordalMPerfect}
Let $G$ be a chordal graph of order $2n$.
Then, $G \in \textsc{Comp-Sub}(\mathcal{PM})$ if and only if $G = K_n\overline{K}_n$.
\end{proposition}

\section{Results on some $P_k$-free graphs}
\label{sec:results-p4s}

In this section, we still consider $\Pi=\mathcal{PM}$ as the property that considers $M$ as a perfect matching.
We begin by showing how to solve $\textsc{Comp-Sub}(\mathcal{PM})$ in polynomial time when the input graph $G$ is $P_5$-free.

\begin{theorem}\label{theo:p5free}
$\textsc{Comp-Sub}(\mathcal{PM})$ is polynomial-time solvable on $P_5$-free graphs.
\end{theorem}

\begin{proof}
Let $G$ be a $2n$-vertex $P_5$-free graph. Recall that if $G \in \textsc{Comp-Sub}(\mathcal{PM})$, then $G$ is decomposable into complementary subgraphs $G_1$ and $G_2$, such the edge cut $M$ of the decomposition is a perfect matching.
Since $G$ is $P_5$-free, the existence of $M$ implies that $G_1$ and $G_2$ are $P_4$-free, that is, $G_1$ and $G_2$ are cographs. Then, the conclusion follows by applying Lemma~\ref{lemma:H_in_cograph}.
\qed \end{proof}

A graph is \textit{extended $P_4$-laden} if every induced subgraph with at most six vertices that contains more than two induced $P_4$'s is $\{2K_2, C_4\}$-free. Extended $P_4$-laden graphs generalize cographs, $P_4$-sparse, $P_4$-lite, $P_4$-laden and $P_4$-tidy graphs, and they were considered under the perspective of partitioning. For instance, Bravo et al.~\cite{bravo2012partitioning} show that partitioning an extended $P_4$-laden graph into at most $k$ independent sets and at most $\ell$ cliques is linear-time solvable, for $k, \ell \geq 1$ and Bravo et~al.~\cite{bravo2013clique} show a linear time algorithm for recognizing graphs that can be partitionable into a clique and a forest. 
In addition, Pedrotti and De Mello~\cite{pedrotti2012minimal} describe a linear-time algorithm that lists the minimal separators of extended $P_4$-laden graphs.

Another related result to partitioning is implied by considering that extended $P_4$-laden graphs are $P_6$-free. The result on $3$-colorability by Randerath and Schiermeyer~\cite{randerath20043} implies that the problem of partitioning a graph into $3$ independent sets is polynomial-time solvable on extended $P_4$-laden graphs.

We present in Proposition~\ref{prop:extendedP4laden} a characterization concerned to $\textsc{Comp-Sub}(\mathcal{PM})$ on extended $P_4$-laden graphs.

\begin{proposition}\label{prop:extendedP4laden}
Let $G$ be an extended $P_4$-laden graph of order $2n$.
It holds that $G \in \textsc{Comp-Sub}(\mathcal{PM})$ if and only if $G = K_n\overline{K}_n$.
\end{proposition}

\begin{proof}
Let $G = K_n\overline{K}_n$.
We analyse the subgraphs of $G$ with at most $6$ vertices to show that $G$ is an extended $P_4$-laden graph.
Let $G'$ be a subgraph of $G$ such that $|V(G')| \leq 6$. If $G'$ is a subgraph of $K_n$ or $\overline{K}_n$, it is clear that $G'$ does not have induced $P_4$'s. Then, we suppose that $V(G')$ intersects both $V(K_n)$ and $V(\overline{K}_n)$. Notice that two induced $P_4$'s arise in $G'$ only if $|V(G') \cap V(\overline{K}_n)| \geq 3$ and $|V(G') \cap V(K_n)| \geq 3$. Since $G$ is a split graph, $G'$ is also a split graph. This implies that $G'$ is $\{2K_2, C_4\}$-free and hence, $G$ is extended $P_4$-laden.

Now, we show that $G \in \textsc{Comp-Sub}(\mathcal{PM})$ implies that $G = K_n\overline{K}_n$.
Suppose that  $G \in \textsc{Comp-Sub}(\mathcal{PM})$, and, by contradiction, that $G \neq K_n\overline{K}_n$.
Since $G \in \textsc{Comp-Sub}(\mathcal{PM})$ there exist a complementary decomposition $(G_1, G_2)$ of $G$, such that the edge cut $M$ of the decomposition is a perfect matching. Let $M = \{ u_1v_1, \dots, u_nv_n \}$ where $u_i \in V(G_1)$ and $v_i \in V(G_2)$, for every $i \in [n]$.

Given that $G \neq K_n\overline{K}_n$, let $u_1u_2u_3$ be a $P_3$ in $G_1$ and $G' = G[\{u_i,v_i : i \in [3]\}]$.

Since $u_iv_i \in E(G')$, for every $i \in [3]$, we have that $\{u_1,v_1,u_3,v_3\}$ induces a $2K_2$ in $G'$. Then, we may suppose that $v_1v_3 \in E(G')$. Notice that $\{u_2,v_2,u_1,v_3\}$ induces a $2K_2$ in $G'$, then we consider that $v_1v_2 \in E(G')$ or $v_2v_3 \in E(G')$. In both possibilities we have an induced $C_4$ in $G'$, by $\{u_1,v_1,u_2,v_2\}$ in the first, and by $\{u_2,v_2,u_3,v_3\}$ in the latter, a contradiction.
\qed \end{proof}

Our last result characterizes cographs \emph{yes}-instances of $\textsc{Comp-Sub}(\mathcal{PM})$. Recall that a \textit{cograph} is a $P_4$-free graph.

\begin{proposition}\label{prop:cographMPerfect}
Let $G$ be a cograph of order $2n$.
Then,  $G \in \textsc{Comp-Sub}(\mathcal{PM})$ if and only if $G = K_2$.
\end{proposition}

\section{Concluding Remarks}
\label{sec:conclusions}

We have considered $\textsc{Comp-Sub}(\mathcal{PM})$ problem when $\mathcal{PM}$ states the edge cut of the decomposition as a perfect matching. We have presented polynomial-time algorithms for solving $\textsc{Comp-Sub}(\mathcal{PM})$ when the input graph $G$ is $hole$-free or $P_5$-free and we have shown
characterizations on chordal, distance-hereditary, and extended $P_4$-laden graphs.

With respect to complexity results, despite its resemblance with the $\NP$-complete problem \textsc{Perfect Matching Cut}, we show that $\textsc{Comp-Sub}(\mathcal{PM})$ is $\GI$-hard when the given input graph $G$ is $\{C_{k\geq 7}, \overline{C}_{k\geq 7} \}$-free. 

We remark that our results by Theorem~\ref{theo:perfHardness} and Theorem~\ref{theo:holefree} address the cases when $G$ is a $C_{\ell \geq k}$-free graph, for every $k \geq 3$, except for $k = 6$. Then, we leave the following conjecture.

\begin{conjecture}
$\textsc{Comp-Sub}(\mathcal{PM})$ is $\GI$-hard on $C_{k\geq 6}$-free graphs.
\end{conjecture}

We also leave the complexity of $\textsc{Comp-Sub}(\mathcal{PM})$ on $C_5$-free graphs and $P_6$-free graphs open.
Furthermore, we still do not know whether $\textsc{Comp-Sub}(\mathcal{PM})$ is $\GI$-complete. 

\bibliographystyle{splncs04}
\bibliography{referencesII}

%

\end{document}